\documentclass[preprint,5p,times,twocolumn,sort&compress]{elsarticle} 
\usepackage{lineno,hyperref}

\usepackage[cmex10]{amsmath}
\usepackage{amssymb}
\usepackage{mathrsfs}
\usepackage{amsthm}

\usepackage{graphicx}
\usepackage{color}
\usepackage{algorithm}
\usepackage{algorithmic}
\usepackage{bm}
\usepackage{setspace}
\hypersetup{hidelinks}

\usepackage{comment}
\usepackage{ifthen}
\newboolean{showcomments}
\setboolean{showcomments}{true}

\allowdisplaybreaks
\journal{Automatica}

\theoremstyle{definition}
\newtheorem{thm}{Theorem}
\newtheorem{lem}{Lemma}
\newtheorem{cor}{Corollary}

\theoremstyle{definition}
\newtheorem{defn}{Definition}
\newtheorem{rem}{Remark}
\newtheorem{cond}{Condition}

{\newtheorem{assum}{\textit{Assumption}}}

\begin{document}
\graphicspath{{Paper_Fig/}}
\setstretch{1}
\begin{frontmatter}
\title{On Nash-Stackelberg-Nash Games under Decision-Dependent Uncertainties: \\ Model and Equilibrium}

\tnotetext[mytitlenote]{This work was supported by the Joint Research Fund in Smart Grid under cooperative agreement between the National Natural Science Foundation of China (NSFC) and State Grid Cooperation of China (No.1966601). Corresponding author: Feng Liu.}

\author[thu]{Yunfan~Zhang}
\author[thu]{Feng~Liu\corref{mycorrespondingauthor}}
\author[sju]{Zhaojian Wang}
\author[cuh]{Yue Chen}
\author[cepri]{Shuanglei Feng}
\author[tud]{Qiuwei Wu}
\author[hku]{Yunhe Hou}

\cortext[mycorrespondingauthor]{Corresponding author}

\address[thu]{State Key Laboratory of Power Systems, Department of Electrical Engineering, Tsinghua University, Beijing, China}
\address[sju]{MOE Key Laboratory of System Control and Information Processing, Department of Automation, Shanghai Jiao Tong University, Shanghai, China}
\address[cuh]{Department of Mechanical and Automation Engineering, The Chinese University of Hong Kong, Hong Kong SAR, China}
\address[cepri]{Renewable Energy Research Center, China Electric Power Research Institute, Beijing, China}
\address[tud]{Center for Electric Power and Energy, Department of Electrical Engineering, Technical University of Denmark, Kgs. Lyngby, Denmark}
\address[hku]{Shenzhen Institute of Research and Innovation, The University of Hong Kong, Hong Kong SAR, China}

\begin{abstract}
In this paper, we discuss a class of two-stage hierarchical games with multiple leaders and followers, which is called Nash-Stackelberg-Nash (N-S-N) games. Particularly, we consider N-S-N games under decision-dependent uncertainties (DDUs). DDUs refer to the uncertainties that are affected by the strategies of decision-makers and have been rarely addressed in game equilibrium analysis. In this paper, we first formulate the N-S-N games with DDUs of complete ignorance, where the interactions between the players and DDUs are characterized by uncertainty sets that depend parametrically on the players' strategies. Then, a rigorous definition for the equilibrium of the game is established by consolidating generalized Nash equilibrium and Pareto-Nash equilibrium. Afterward, we prove the existence of the equilibrium of N-S-N games under DDUs by applying Kakutani's fixed-point theorem. Finally, an illustrative example is provided to show the impact of DDUs on the equilibrium of N-S-N games.
\end{abstract}

\end{frontmatter}

\section{Introduction}
\subsection{Background}
Game theory provides a powerful tool to deal with decision-making problems with multiple players in various disciplines, ranging from economy, military, politics, to social science and engineering. Hierarchical games model sequential decision-making problems and can be divided into two-stage hierarchical games and multi-stage ones, according to the number of decision stages or levels. The two-stage hierarchical game with multiple leaders and followers is called Nash-Stackelberg-Nash (N-S-N) game in this paper.

The existing works on hierarchical games are typically established in deterministic conditions. However, uncertainties, some of which are decision-dependent, inherently exist in many real-world decision-making problems. For example, the uncertain damages from climate change have a dependency on the global climate policy \cite{Webster2012An}. Another example is the demand-response program on buildings' electricity consumption where the reserve demand from the system operator is endogenously uncertain due to its dependency on the building's reserve capacity provided in the day-ahead market \cite{zhang2017}. Uncertainties that can be affected by decisions are called decision-dependent uncertainties (DDUs) or endogenous uncertainties interchangeably. Distinguished from DDUs, decision-independent uncertainties (DIUs), or called exogenous uncertainties, are not affected by the decision-makers' actions.

Fruitful works have addressed the optimization problems under 
DDUs \cite{Nohadani2018,Lappas2018,mypaper2021}. Consideration of DDUs in N-S-N game equilibrium analysis, however, is challenging. The multiple players in an N-S-N game whose payoffs are affected by some uncertain factors, can in turn exercise proactive control on the uncertainties cooperatively or non-cooperatively. In this regard, how to define the equilibrium of an N-S-N game under DDUs and justify its existence appear to be crucial albeit difficult, and have not been addressed to the best of the authors' knowledge.

\subsection{Literature Review}
In the literature, two kinds of epistemic states characterize the uncertainties in games.

\noindent$\cdot$\textit{\textbf{ Partial Ignorance.}} Ignorance is a state of knowledge characterized by knowing nothing or having no reliable information about the matter of interest \cite{2015Decision}. The uncertain factor of partial ignorance is assigned a probability distribution or a set of possible probability distributions that can be incorporated into appropriate mathematical programming models. Motivated by the variety of ideas in stochastic programming, the players in game problems may choose to optimize the expected payoff, the risk measures, or the mean-risk composite indexes to obtain the most favorable outcome in accordance with their risk preferences, see \cite{2017Two} for a review. 


\noindent$\cdot$\textit{\textbf{Complete Ignorance.}} Uncertainty is said to be of complete ignorance when the probability distribution of the uncertain factor is difficult to obtain or even unavailable \cite{Giang2016}. Though reliable evidence is absent, experts succeed in giving a deterministic region called \textit{uncertainty set}, within which all possible uncertainty realizations stay. The sustained interest in decision-making in situations of complete ignorance is motivated by many real-world circumstances: adequate historical data, which are necessary for an explicit and credible probability distribution of the unknown factor, are rarely publicly available.

Analogous to robust optimization \cite{ben2009robust} where system's performance under worst-case uncertainty realization is guaranteed, only the extreme consequence caused by the uncertainty matters in an uncertain game problem under complete ignorance. In this sense, the players' best response under uncertainty is stipulated to be the strategy with the best \textit{worst-case performance}; and the uncertain factor is regarded as a \textit{virtual player} standing up to the real players whose strategy set is the uncertainty set that contains all its possible realizations. The real players' interest in hedging against the risks of uncertainty is intrinsically captured by the virtual player's ambition of deteriorating the payoffs of real players. Such a combination of the idea of robust optimization and the concept of Nash equilibrium (NE) \cite{Nash:Noncooperative} renders a distribution-free equilibrium concept called robust-Nash equilibrium (RNE) \cite{Aghassi2006}. The RNE has been studied by \cite{Aghassi2006} for N-person non-cooperative finite games, \cite{Hayashi2005} for bimatrix games, \cite{Nishimura2009,Nishimura2012} for normal-form N-person non-cooperative infinite games, \cite{YuHRobust2013} for multi-objective games, and \cite{2013EXISTENCE,2015MULTI} for leader-follower games.

The aforementioned study of RNE focus on the situations of local uncertainty and the worst case is taken with respect to each player's payoff individually. To further extend the idea of robust game to the situations of shared uncertainty sources, subsequent literature has built novel equilibrium concepts upon the insights of Pareto optimality. Zhukovskii and Tchikry \cite{Zhukovskii} first propose the concept of Nash-Slater equilibrium (NS-equilibrium) for uncertain N-person non-cooperative games, combining the concept of classic NE and Pareto optimality in multi-objective optimization theory. Specifically, the pessimistic and conservative game players assume that the virtual player (the shared uncertainty sources) tries his best to make the payoffs of all real players worse off simultaneously. As such, the worst-case performances of the real players' strategies are given collectively by considering the multi-objective optimization problem faced by the virtual player; and on the NS-equilibrium, the virtual player's strategy is the solution to his corresponding multi-objective optimization problem in the sense of Pareto optimality. Then, the study of NS-equilibrium is followed by Larbani and Lebbah \cite{LARBANI:Aconcept} by introducing the concept of zero-sum equilibrium (ZS-equilibrium) where the Pareto efficiency of real players' strategies is additionally considered. Yang and Pu \cite{ref7} extend the previous work to uncertain leader-follower games, by considering the collective interest of leaders hedging against the risks of uncertainty; then, Zhang, et al \cite{ref10} discuss the existence of equilibrium in the case of one leader and multiple followers and take into account the risk awareness of the followers. Moreover, Nessah, et al \cite{Nessah2015} have extended the notion of ZS-equilibrium to a Coalitional ZP-equilibrium where an internal coalition stability condition for a given coalition structure is added. 

Despite the fruitful work extending the concept of RNE to more complex game situations, they deal with problems where only DIUs are involved and the chosen strategies of players have no effect on the uncertain resolution. Recent advances \cite{zhang2017,Lappas2018,Nohadani2018,mypaper2021} extending robust optimization to dealing with DDUs may provide insights into the mathematical characterization of DDUs with complete ignorance: the conventional `static' uncertainty sets are extended to set-valued maps parameterized in the decision variables. From the game theoretic perspective, the real players and the virtual player (the DDUs) interact not only at the level of payoffs but also at the level of strategy sets, which is similar to the notion of generalized uncertainty introduced for N-person non-cooperative game in \cite{ref8,ref9}. However, recall our goal of this paper, the consideration of DDUs or the so-called generalized uncertainties in hierarchical games would raise more technical challenges for equilibrium analysis, especially when the DDUs are ambiguous on the leader's moves but revealed prior to the followers' moves. To the best of the authors' knowledge, DDUs in leader-follower games have not been investigated.

\subsection{Contribution and Organization}
This paper aims to mathematically formulate the  N-S-N game with DDUs of complete ignorance, and characterize the existence of its Nash equilibrium. Specifically, the following three key issues are addressed: 

(i) We establish the normal-form N-S-N game model with DDUs incorporated and characterize the interaction between the players and the shared uncertain factors at both the level of payoffs and strategy sets. The N-S-N game model under DDUs differs from the existing works \cite{Zhukovskii,LARBANI:Aconcept,ref8,ref9,Nessah2015,ref7,ref10}. On the one hand, our model with hierarchical structure is not implied by the uncertain non-cooperative games in \cite{Zhukovskii,LARBANI:Aconcept,ref8,ref9,Nessah2015} due to the non-uniqueness of the generalized Nash equilibrium (GNE) of the follower-level game. On the other hand, unlike the uncertain leader-follower games in \cite{ref7,ref10}, cognitive differences of the players (featured as leaders and followers) to the uncertainty resolution are specified as the DDUs are ambiguous on the leaders' moves then revealed prior to the followers' moves, which makes the leader-level equilibrium analysis more challenging.

(\romannumeral2) We rigorously define the equilibrium of N-S-N games under DDUs by consolidating the concepts of GNE, RNE, and weak and non-weak Pareto-Nash equilibrium (PNE). Compared with \cite{Zhukovskii,LARBANI:Aconcept,ref7,ref10,Nessah2015,ref8,ref9} where weakly Pareto efficiency is leveraged to characterize the players' collective interest against the uncertainties, our extensions include both weak and non-weak Pareto optimality and respectively render the weak and strong equilibrium of the studied game. We also emphasize the rationality of the equilibrium by invoking the criteria including 1) feasibility; 2) dominance; 3) robustness; and 4) Pareto axiom.

(\romannumeral3) We prove the existence of the equilibrium of N-S-N games under DDUs by applying Kakutani's fixed-point theorem. It is revealed that the equilibrium existence is established on 1) continuity assumptions; 2) compactness assumptions, and 3) (quasi-) convexity assumptions, which are commonly postulated in non-cooperative games and hierarchical games. Our result is a unification and improvement to the existence theorems for GNE and PNE in the literature. Specifically, the (quasi-) convexity assumptions on payoff functions are milder than the properly-quasiconcavelike conditions in \cite{ref7,ref8,ref9}.

The rest of the paper is organized as follows. Notations and preliminaries are presented in Section \ref{NotationsandPreliminaries}. Section \ref{NSNDDUS} derives an N-S-N game model under DDUs and specifies its equilibrium. Some useful remarks are provided as well. Then in Section \ref{Existence}, our main result, the equilibrium existence theorem is given with proof. In section \ref{mathproblem}, an illustrative example is given to show how DDUs affect the equilibrium of N-S-N games. Section \ref{discussion} concludes the paper.

\section{Notations and Preliminaries}
\label{NotationsandPreliminaries}
\subsection{Notations}

Throughout this paper,  $\mathbb{R}^n$ ($\mathbb{R}^n_+$) is the $n$-dimensional (non-negative) Euclidean space. Unless otherwise specified, we use uppercase letters, for example $X$, to denote non-empty sets in Euclidean space, and ${\rm int}X$  the set of interior points in $X$.  $X\times Y$ denotes the Cartesian product of sets $X$ and $Y$. The sum of two sets refers to $X_1+X_2\triangleq\{x_1+x_2|x_1\in X_1,x_2\in X_2\}$. $X\backslash Y$ denotes the set of elements that belong to $X$ but not to $Y$. Given a collection of points $x_i$ for $i$ in a certain set $N=\{1,...,n\}$, $x\triangleq(x_1,\ldots,x_n)^{\mathsf{T}}$ denotes its assembling form and $x_{-i}\triangleq(x_1,\ldots,x_{i-1},x_{i+1},\ldots,x_n)^{\mathsf{T}}$. Given a collection of sets $X_i$ for $i\in N$, $X\triangleq X_1\times \ldots \times X_n$ and $X_{-i}\triangleq X_1\times\ldots\times X_{i-1}\times X_{i+1}\times \ldots \times X_n$. Let $\left\{x^k\right\}\rightarrow x$ denote the limit of a point sequence $\left\{x^k\right\}$.

$\mathcal{F}:X\rightrightarrows Y$ denotes a set-valued map if $\mathcal{F}(x)$ is a non-empty subset of $Y$ for all $x\in X$. The graph of $\mathcal{F}$ is defined as $\text{graph}\mathcal{F}\triangleq\left\{(x,y)\in X\times Y|y\in \mathcal{F}(x)\right\}$. $\mathcal{F}$ is convex iff $\text{graph}\mathcal{F}$ is a convex set. $\mathcal{F}$ is closed iff $\text{graph}\mathcal{F}$ is a closed set. We will establish our work on the semicontinuity property of set-valued maps, which are well-documented in \cite{Aubin:Differential}.

\subsection{Non-cooperative Games}
In this subsection, we start with the classic N-person non-cooperative game, then briefly review four typical extensions of non-cooperative games: (\romannumeral1) generalized game; (\romannumeral2) multi-objective game; (\romannumeral3) uncertain game; and (\romannumeral4) hierarchical game, for better understanding the more complicated N-S-N game under DDUs in this paper. 

\noindent$\cdot$\textit{\textbf{{Non-cooperative Game}:}} $\Theta\triangleq\{N,X_i,f_i\}$ is an N-person non-cooperative game where $N=\{1,...,n\}$ denotes the set of $n$ players. For any $i\in N$, $X_i$ is the strategy space of player $i$ and $f_i:X\rightarrow \mathbb{R}^1$ is the payoff function of player $i$. $x^*\in X$ is the NE \cite{Nash:Noncooperative} of $\Theta$ if for any $i\in N$,
\begin{eqnarray}
\notag
f_i(x_i^*,x_{-i}^*)=\max_{x_i}f_i(x_i,x_{-i}^*)\ {\rm s.t.}\ x_i\in X_i.
\end{eqnarray}
Next, we show how $\Theta$ is extended to four different scenarios and the corresponding variants of NE.

\noindent$\cdot$\textit{\textbf{{Generalized Game}:}} In a generalized game, considering shared constraints among players, each player's strategy belongs to a so-called feasible strategy set that depends upon his rivals' strategies. Specifically, $\Theta^{\rm G}\triangleq\{N,\mathcal{X}_i,f_i\}$ is a generalized N-person non-cooperative game where $N$ is the set of players, $\mathcal{X}_i(x_{-i})$ is the feasible strategy set of player $i$ and $f_i$ is the payoff of player $i$. $x^*$ is called the GNE \cite{Facchinei2007} of $\Theta^{\rm G}$ if for any $i\in N$, 
\begin{eqnarray}
\notag
x_i^*\in \mathcal{X}_i(x_{-i}^*)
\end{eqnarray}
and 
\begin{eqnarray}
\notag
f_i(x_i^*,x_{-i}^*)=\max_{x_i}f_i(x_i,x_{-i}^*)\ {\rm s.t.}\ x_i\in \mathcal{X}_i(x_{-i}^*).
\end{eqnarray}

\noindent$\cdot$\textit{\textbf{{Multi-objective Game}:}} A multi-objective game refers to the game with vector payoffs. $\Theta^{\rm MO}\triangleq\{N,X_{i},f_i\}$ is a multi-objective N-person non-cooperative game where $f_i\triangleq\{f_i^1,\ldots,f_i^{k_i}\}:X\rightarrow \mathbb{R}^{k_i}$ is the vector-valued payoff function of player $i$. $x^*$ is called the weak PNE \cite{YuJ:Thestudy} of $\Theta^{\rm MO}$ if for any $i\in N$ and $x_{i}\in X_{i}$,
\begin{eqnarray}
\notag
f_i(x_i,x_{-i}^*)-f_i(x_i^*,x_{-i}^*)\notin {\rm int}\mathbb{R}_{+}^{k_i},
\end{eqnarray}
which implies that there exists $k\in \{1,...,k_i\}$ such that
\begin{eqnarray}
\notag
f_i^k(x_i,x_{-i}^*)-f_i^k(x_i^*,x_{-i}^*)\le 0.
\end{eqnarray}
$x^*$ is the PNE of $\Theta^{\rm MO}$ if for any $i\in N$ and $x_{i}\in X_{i}$,
\begin{eqnarray}
	\notag
	f_i(x_i,x_{-i}^*)-f_i(x_i^*,x_{-i}^*)\notin\mathbb{R}_{+}^{k_i}\backslash\{0\},
\end{eqnarray}
which implies that either
\begin{eqnarray}
\notag
f_i(x_i,x_{-i}^*)-f_i(x_i^*,x_{-i}^*)=0,
\end{eqnarray}
or there exists $k\in \{1,...,k_i\}$ such that
\begin{eqnarray}
\notag
f_i^k(x_i,x_{-i}^*)-f_i^k(x_i^*,x_{-i}^*)< 0.
\end{eqnarray}

\noindent$\cdot$\textit{\textbf{{Uncertain Game}:}} $\Theta^{\rm U}\triangleq\{N,X_{i},f_i,W\}$ is an uncertain N-person non-cooperative game with undeterministic factor $w$ affecting the payoffs to players and $W$ is the domain that $w$ can vary within. $f_i:X\times W\rightarrow \mathbb{R}^1$ is the payoff to player $i$. Let $f\triangleq(f_1,\ldots,f_n)^{\mathsf{T}}$ be the collection of all players' payoffs. $(x^*,w^*)$ is the NS-equilibrium \cite{Zhukovskii} of $\Theta^{\rm U}$ if (i) Given $w^*$, $x^*$ is the NE; and (ii) Given $x^*$, for any $w\in W$, 
\begin{eqnarray}
\notag
f(x^*,w^*)-f(x^*,w)\notin {\rm int}\mathbb{R}_{+}^{n},
\end{eqnarray}
which implies that there exists $i\in N$ such that
\begin{eqnarray}
\notag
f_i(x^*,w^*)- f_i(x^*,w)\le 0.
\end{eqnarray}

\noindent$\cdot$\textit{\textbf{(Two-Level) Hierarchical Game:}}
Two-level hierarchical games, dealing with two-stage sequential decision problems involving multiple leaders and followers, can be denoted by $\Theta^{\rm H}\triangleq\{N,X_i,f_i,M,\mathcal{Y}_j,\phi_j\}$
where $N$ is the set of $n$ leaders and $M$ is the set of $m$ followers. $f_i:X\times Y\rightarrow \mathbb{R}^1$ is the payoff of leader $i$ and $\phi_j:X\times Y\rightarrow \mathbb{R}^1$ is the payoff of follower $j$. $\mathcal{Y}_j:X\times Y_{-j}\rightrightarrows Y_{j}$ is the feasible strategy set of follower $j$. Each follower intends to maximize his payoff $\phi_j$ by solving
\begin{eqnarray}
\notag
\begin{split}
\mathcal{G}_j(x,y_{-j})=\arg &\max_{y_j}\phi_j(x,y_{j},y_{-j})\\
&\ {\rm s.t.}\ y_j\in \mathcal{Y}_j(x,y_{-j}).
\end{split}
\end{eqnarray}
Denote by $\mathcal{G}(x)\triangleq\{y|y_j\in\mathcal{G}_j(x,y_{-j}),\forall j\in M\}$
the \textit{reaction map} of all followers. $x^*$ is the N-S-N equilibrium of $\Theta^{\rm H}$ \cite{2005Quasi} if for any $i\in N$, there exists $y^{[i]*}$ with the same dimension of $y$ such that
\begin{eqnarray}
\setlength{\abovedisplayskip}{0cm}
\setlength{\belowdisplayskip}{0cm}
\notag
\begin{split}
(x_i^*,y^{[i]*})\in\arg &\max_{x_i,y} f_i(x_i,x^*_{-i},y)\\
&\ {\rm s.t.}\ x_i\in X_i,y\in \mathcal{G}(x_i,x^*_{-i}).
\end{split}
\end{eqnarray}
By this definition, no leader can increase his payoff by a unilateral change of his strategy, thus $x^*$ is also the NE of the $n$ leaders. Given $x^*$, for any $i\in N$, $y^{[i]*}\in \mathcal{G}(x^*)$ is the GNE of the $m$ followers, indicating that no follower can promote his payoff by a unilateral change of his strategy.

\subsection{Key Lemmas}
Next, several key lemmas are introduced to develop the theoretical results in this paper. First, we introduce the concepts of \emph{marginal function} and \emph{marginal map}, as well as their continuity property, in the following lemma.

\begin{lem}[Continuity of marginal function and marginal map \cite{Aubin:Differential}]
	\label{lemma:marginal}
	Suppose $\mathcal{G}:Y\rightrightarrows X$ is a set-valued map and $f:X\times Y\rightarrow \mathbb{R}^1$ is a real-valued function defined on $X\times Y$. If $f$ is continuous on $X\times Y$ and $\mathcal{G}$ is continuous with compact values, then 
	
	(i) The marginal function
	\begin{eqnarray}
	\notag
	M(y)\triangleq\sup\nolimits_{x}f(x,y)\ {\rm s.t.}\ x\in \mathcal{G}(y)
	\end{eqnarray}
	is continuous with respect to $y$; 
	
	(ii) The marginal set-valued map 
	\begin{eqnarray}
	\notag
	\mathcal{M}(y)\triangleq\{x\in \mathcal{G}(y)|M(y)=f(x,y)\}
	\end{eqnarray}
	which denotes the solutions to the maximization problem $M(y)$, is upper semi-continuous with respect to $y$.
\end{lem}

Then, Kakutani's fixed-point theorem, which is often applied to prove the existence of NE, is given below.
\begin{lem}[Kakutani's fixed-point theorem \cite{Kakutani:AGeneralization}]
	\label{kakutani}
	Assume $X$ be a non-empty, compact and convex subset of Euclidean space $\mathbb{R}^n$. $\mathcal{F}:X\rightrightarrows X$ is an upper semi-continuous set-valued map on $X$ with compact convex values. Then  there exists $x^*\in X$ such that $x^*\in\mathcal{F}(x^*)$.
\end{lem}

\section{N-S-N Game under DDUs and Its Equilibrium}
\label{NSNDDUS}
\subsection{Game Model}

In this subsection, we establish the model of Nash-Stackelberg-Nash games under DDUs. Four key features of this class of games are taken into account:

(\romannumeral1) There exist $n$ leaders and $m$ followers, as the players of the game. Each player has his own pending strategy and payoff function, and has full authority to act individually to maximize his payoff. The payoff of a player is contingent on not only his own strategy, but also strategies of the other players in the game. Coalitions and re-distribution of payoffs are not allowed, indicating the scope of non-cooperative games. 

(\romannumeral2) A two-stage decision process is involved where the leaders move observably first and then the followers act sequentially in response to the leaders' actions. The leaders also know beforehand that the followers would observe their actions.

(\romannumeral3) There exists an uncertain factor in the game, the value of which is ambiguous during the first-stage (when the leaders move) and is revealed at the beginning of the second stage (prior to the followers' move). The uncertain parameter has an effect on the payoffs of both leaders and followers, as well as the strategy sets of followers.

(\romannumeral4) The uncertain factor is decision-dependent but of complete ignorance knowledge. Specifically, leaders can exercise proactive control to limit the range that the uncertain factor varies within. Apart from this, leaders have no reliable information about the uncertain factor, such as the probability distribution.

A diagram of considered game is presented in Figure \ref{nsn}. With the above settings, the normal-form of the game is given as follows:
\begin{figure}[!htp]
	\centering
	\includegraphics[width=0.45\textwidth]{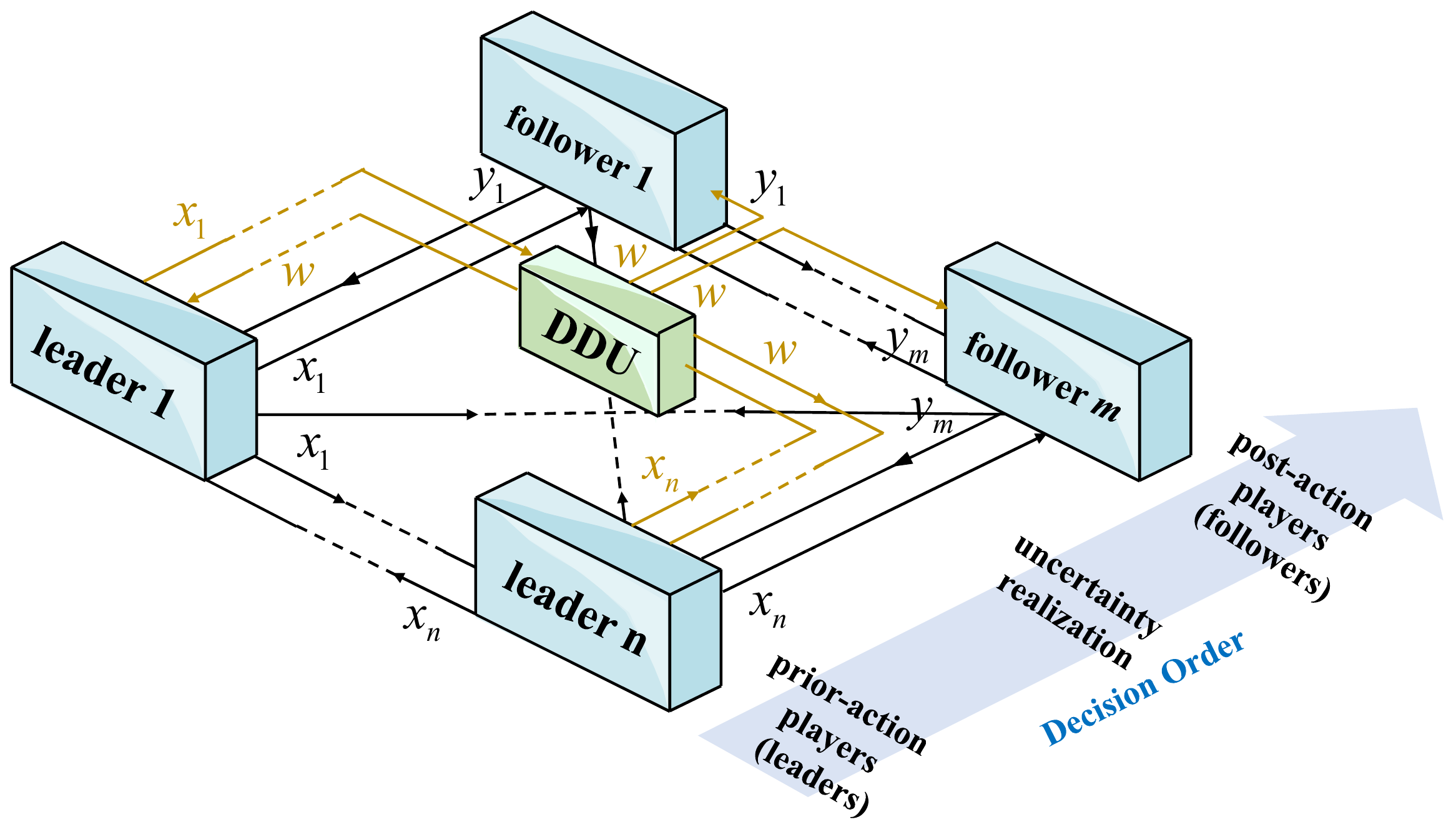}
	\caption{The diagram of game $\Theta^{\rm DDU}$.}
	\label{nsn}
\end{figure}

\begin{defn}
	\label{def:model}
	A Nash-Stackelberg-Nash game under DDUs is defined as an eight-tuple
	\begin{eqnarray}
	\label{def_model}
	\Theta^{\rm DDU}\triangleq\{N,X_i,\mathcal{X}_i,Y,W,\mathcal{W},f_i,\mathcal{G}\}
	\end{eqnarray}
	the elements involved in which are explained as below.
		
    \noindent $\cdot$ \textbf{Leaders:}  $N\triangleq\{1,\ldots,n\}$ is the set of the $n$ leaders. For any $i\in N$, $x_i$ is the strategy of leader $i$ and $X_i$ is the strategy space of $x_i$. Considering shared constraints among leaders, the feasible strategy set of leader $i$ is denoted by $\mathcal{X}_i:X_{-i}\rightrightarrows X_i$. Single-valued function $f_i:X\times Y\times W\rightarrow\mathbb{R}^1$ is the payoff of leader $i$ that he or she wishes to maximize. The vector-valued function $f\in\mathbb{R}^n$ consisting of the payoffs of all leaders is $f\triangleq(f_1,\ldots,f_n)^{\mathsf{T}}$.

	\noindent$\cdot$ \textbf{Uncertainties:} $w$ denotes the collection of all uncertain factors and $W$ is the uncertainty set that contains all possible realizations of $w$. Furthermore, set-valued map $\mathcal{W}:X\rightrightarrows W$ is utilized to characterize the dependency of $w$ on decision $x$. Once the strategies of leaders are determined, the realization of $w$ must lie within $\mathcal{W}(x)$.
		
	\noindent$\cdot$ \textbf{Followers:} $y$ collectively denotes strategies of the followers and $Y$ is the corresponding strategy space. The followers' problem is modeled as a generalized Nash game, which is parameterized by the leaders' strategies $x$ and the revealed value of $w$. Let $\mathcal{G}: X\times W \rightrightarrows Y$ denote the set of GNE of the followers. $\mathcal{G}(x,w)$ is also the reaction map of the followers to which $(x,w)$ are exogenous. Note that $\mathcal{G}(x,w)$ is not necessarily a singleton.
		
	\noindent$\cdot$ \textbf{Decision Order:} The involved two groups of the players in the game and their action sequence are: the leaders first and then the followers. Post-action players are in view of the moves of prior-action players; and the prior-action players know ex-ante that the post-action players observe their actions and are able to anticipate the response of the post-action players. The uncertainty is realized after the decision of leaders and before the actions of followers. 
\end{defn}

Game $\Theta^{\rm DDU}$ is assumed to be under complete information, i.e., the description of the game, regarding the set of players, the strategy sets, and the utility functions, is common knowledge among the players. 

Game $\Theta^{\rm DDU}$ is an extended form of non-cooperative games by combining the characteristics of game $\Theta^{\rm G}$, $\Theta^{\rm H}$, and $\Theta^{\rm U}$. Analogous to game $\Theta^{\rm G}$, the strategy of one player belongs to the set that explicitly depends on the strategies of the other players. Analogous to game $\Theta^{\rm H}$, players are in a position of the sequentially different levels. Analogous to game $\Theta^{\rm U}$, undetermined parameters are involved in game $\Theta^{\rm DDU}$ in the case of complete ignorance, i.e., the players in the game are only aware of the domain where the uncertain parameters vary. In particular, the concept of DDUs is introduced to game $\Theta^{\rm DDU}$ so that the players plagued by uncertainties in turn have an effect on the domain within which the uncertain parameter varies. Regarding the DDU parameter $w$, we have the following remark.

\begin{rem}[The DDUs] (i) The decision-dependency of $w$ is characterized by the decision-dependent uncertainty set $\mathcal{W}(x)$ without an explicit form. This generic set-valued model covers distinguishable formulations of DDU set in existing works\cite{zhang2017,Lappas2018,Nohadani2018,mypaper2021}. Also, $\mathcal{W}(x)$ applies readily with DIU sets by setting $\mathcal{W}(x)=W$ for any $x\in X$. (ii) In game $\Theta^{\rm DDU}$, all the uncertain factors are collectively denoted by $w$ and $w$ is a shared parameter for all leaders. It readily applies to the case that each leader has his own local uncertain factor. 
\end{rem}

\subsection{Equilibrium}
In this subsection, we define the variant of NE for the game $\Theta^{\rm DDU}$ and explain its rationality. First, we introduce Condition \ref{cond:equilibrium} which is necessary for defining the equilibrium of $\Theta^{\rm DDU}$. Then the equilibrium of $\Theta^{\rm DDU}$ is given in Definition \ref{def:eq}.

\begin{cond}
\label{cond:equilibrium}
Consider a point $(x^*,w^*)\in X\times W$.

(a) For any $i\in N,x_i^*\in\mathcal{X}_i(x_{-i}^*)$;

(b) $w^*\in \mathcal{W}(x^*)$; 

(c) For any $i\in N$, there exists $y^{[i]*}\in \mathcal{G}(x^*,w^*)$ such that
\begin{eqnarray}
\notag
f_i(x^*,y^{[i]*},w^*)\ge f_i((x_i,x_{-i}^*),y,w^*)
\end{eqnarray}
holds for any $x_i\in \mathcal{X}_i(x_{-i}^*)$ and $y\in \mathcal{G}((x_i,x_{-i}^*),w^*)$;

(d1) There exists $y^{[w]*}\in \mathcal{G}(x^*,w^*)$ such that
\begin{eqnarray}
\notag
f(x^*,y^{[w]*},w^*)-f(x^*,y,w)\notin {\rm int }\mathbb{R}^n_+
\end{eqnarray}
holds for any $w\in \mathcal{W}(x^*)$ and $y\in \mathcal{G}(x^*,w)$;

(d2) There exists $y^{[w]*}\in \mathcal{G}(x^*,w^*)$ such that
\begin{eqnarray}
	\notag
	f(x^*,y^{[w]*},w^*)-f(x^*,y,w)\notin \mathbb{R}^n_+\backslash\{0\}
\end{eqnarray}
holds for any $w\in \mathcal{W}(x^*)$ and $y\in \mathcal{G}(x^*,w)$.
\end{cond}

\begin{defn}[Equilibrium of Game $\Theta^{\rm DDU}$]
\label{def:eq} Consider a point  $(x^*,w^*)\in X\times W$.

(a) If $(x^*,w^*)$ satisfies Condition \ref{cond:equilibrium} (a)-(c) and (d1), then $(x^*,w^*,\mathbf{y}^*)$ is a weak equilibrium point of game $\Theta^{\rm DDU}$ where
\begin{eqnarray}
\notag
\mathbf{y}^*\triangleq(y^{[1]*},...,y^{[n]*},y^{[w]*})
\end{eqnarray}
is the collection of anticipations on followers' reactions originating from Condition \ref{cond:equilibrium} (c) and (d1).

(b) If $(x^*,w^*)$ satisfies Condition \ref{cond:equilibrium} (a)-(c) and (d2), then $(x^*,w^*,\mathbf{y}^*)$ is a strong equilibrium point of game $\Theta^{\rm DDU}$ where
\begin{eqnarray}
	\notag
	\mathbf{y}^*\triangleq(y^{[1]*},...,y^{[n]*},y^{[w]*})
\end{eqnarray}
is the collection of anticipations on followers' reactions originating from Condition \ref{cond:equilibrium} (c) and (d2).
\end{defn}

The rationality of the equilibrium of $\Theta^{\rm DDU}$ in Definition \ref{def:eq} is explained in the following remark.

\begin{rem}[Rationality of the Equilibrium]
	\textcolor{white}{text}
	
	\noindent$\cdot$ \textbf{(i) Feasibility:} Condition \ref{cond:equilibrium} (a) and (b) ensure the feasibility of $x^*$ and $w^*$, respectively. Given the leaders' strategy $x^*$, the realization of uncertainty $w^*$ must stay within $\mathcal{W}(x^*)$. This condition is consistent with the feasibility requirement on the GNE of game $\Theta^{\rm G}$. 
	
	\noindent$\cdot$ \textbf{(ii) Dominance:} Condition \ref{cond:equilibrium} (c) emphasizes that given the realization of uncertainty $w^*$, no leader can gain better profit by unilaterally adjusting his own strategy. Thus,  equilibrium strategies of all leaders, $x^*$, form a GNE that dominates other decisions in $X$. On the other hand, given $x^*$ and $w^*$, the GNE of the follower-level non-cooperative game is characterized by $\mathcal{G}(x^*,w^*)$. Thus, if the uncertain parameter is fixed as $w^*$,  $(x^*,$ $\mathcal{G}(x^*,w^*))$ is the GNE of the nominal two-level hierarchical game $\Theta^{\rm H}$.
	
    \noindent$\cdot$ \textbf{(iii) Robustness:} Motivated by the concept of RNE, the best response of leaders under uncertainties is stipulated as the strategy with the best ``worst-case performance". The uncertain parameter $w$ is assumed to be a virtual player who intends to worsen the payoffs of all leaders simultaneously as much as possible. Distinguished from the classic RNE in \cite{Aghassi2006} where the worst-case performance is taken with respect to each player individually, we follow the insights in the NS-equilibrium \cite{Zhukovskii} of game $\Theta^{\rm U}$ to characterize the collective interest of the leaders against the uncertainty. Thus the virtual player is stipulated to be on the same footing with the leaders and perform a zero-sum game with the leaders.
    
    \noindent$\cdot$ \textbf{(iv) Pareto Axiom:} Motivated by the weak and non-weak PNE of multi-objective game $\Theta^{\rm MO}$, the virtual player's best response is stipulated with the aid of the weak and non-weak Pareto optimality axiom. Given the leaders' strategies $x^*$ and the virtual player's strategy $w^*$, Condition \ref{cond:equilibrium} (d1) states that the virtual player cannot reduce the payoffs of \textit{all} the leaders simultaneously by unilaterally adjusting his own strategy; and Condition \ref{cond:equilibrium} (d2) states that the virtual player cannot reduce the payoff of \textit{any} leader without improving the payoff of any other leader by unilaterally adjusting his own strategy. Condition \ref{cond:equilibrium} (b) and (d1) together imply that $w^*$ is the weak Pareto efficient solution \cite{YuJ:Thestudy} to the following multi-objective optimization problem. 	
    \begin{eqnarray}
    	\label{def:moo}
    	\begin{split}	
    	&\min_{w,y}\ \left[f_1(x^*,y,w),...,f_n(x^*,y,w)
    	\right]^{\mathsf T} \\
    	&{\rm s.t.}\ w\in\mathcal{W}(x^*),y\in\mathcal{G}(x^*,w)
        \end{split}
    \end{eqnarray}
   And Condition \ref{cond:equilibrium} (b) and (d2) together imply that $w^*$ is the Pareto efficient solution to problem \eqref{def:moo}.
\end{rem}
Also note that, in Definition \ref{def:eq}, the followers' equilibrium strategies, $y^{[1]*},...,y^{[n]*}$ and $y^{[w]*}$, while all being elements of the equilibrium response set $\mathcal{G}(x^*,w^*)$, are not required to be identical, providing additional modeling flexibility \cite{2005Quasi}. This is because $y^{[i]*}$ is leader $i$'s anticipation of the followers' response to $(x^*,w^*)$ and $y^{[w]*}$ is the virtual player's anticipation of the followers' response to $(x^*,w^*)$. These anticipations are not necessarily identical when $\mathcal{G}(x,w)$ is not a singleton. Thus, an equilibrium of $\Theta^{\rm DDU}$, as given in Definition \ref{def:eq}, actually contains $n+1$ possible outcomes which are $(x^*,w^*,y^{[1]*}),...,(x^*,w^*,y^{[n]*})$ and $(x^*,w^*,y^{[w]*})$. Such an equilibrium can be reduced to a unique outcome in the following two cases:

\noindent$\cdot$ \textit{Case 1:} $\mathcal{G}$ is a single-valued map. One could define a variation of the follower's problem by stipulating certain schemes or rules, say, market clearing mechanism, to enforce $y^{[1]*}=\ldots=y^{[n]*}=y^{[w]*}$.

\noindent$\cdot$ \textit{Case 2:} Given $(x^*,w^*)$, there exists $y^*\in\mathcal{G}(x^*,w^*)$ such that:

(i) For any $i\in N$ and $x_i\in \mathcal{X}_i(x_{-i}^*)$, and for any $y\in \mathcal{G}((x_i,x_{-i}^*),w^*)$, there is
\begin{eqnarray}
	\notag
	f_i(x^*,y^{*},w^*)\ge f_i((x_i,x_{-i}^*),y,w^*);
\end{eqnarray} 

(ii) For any $w\in \mathcal{W}(x^*)$ and $y\in \mathcal{G}(x^*,w)$, there is
\begin{eqnarray}
	\notag
	f(x^*,y^{*},w^*)-f(x^*,y,w)\notin {\rm int }\mathbb{R}^n_+\ (\text{or }\mathbb{R}^n_{+}\backslash\{0\}).
\end{eqnarray}
In Case 2, the equilibrium $(x^*,w^*,\mathbf{y}^*)$ reduces to $(x^*,w^*,y^*)$.

If neither of the above two cases is satisfied, Definition \ref{def:eq} provides a more generalized and flexible concept of equilibrium.

\section{Existence of the Equilibrium}
\label{Existence}
\subsection{Main Result}
We justify the existence of the equilibrium of game $\Theta^{\rm DDU}$ under the following assumptions.
\begin{assum}
	\label{asmp:SetValuedMap} The following conditions hold. 

	(a) For any $i\in N$, $\mathcal{X}_i:X_{-i}\rightrightarrows X_i$ is continuous set-valued map with non-empty compact convex values;
	
	(b) $\mathcal{G}:X\times W\rightrightarrows Y$ is continuous set-valued map with convex graph and non-empty compact values;
	
	(c) $\mathcal{W}:X\rightrightarrows W$ is continuous set-valued map with non-empty compact convex values.
	
\end{assum}

\begin{assum}
\label{asmp:fi} 
For any $i\in N$, $f_i:(\prod_{i\in N}X_i)\times Y\times W\rightarrow\mathbb{R}^1$ satisfies that

(a) For any $x_{-i}\in X_{-i}$ and $w\in W$, $f_i((x_i,x_{-i}),y,w)$ is quasi-concave with respect to $(x_i,y)$; 

(b) There exists a non-empty subset of $N$, namely, $S\subseteq N, S\neq \emptyset$, such that for any $x\in X$ and $i\in S$, $-f_i(x,y,w)$ is concave with respect to $(y,w)$;

(c) $f_i$ is continuous in $(\prod_{i\in N}X_i)\times Y\times W$.
\end{assum}

\begin{thm}
\label{equilibrium_existence}
Game $\Theta^{\rm DDU}=\{N,X_i,\mathcal{X}_i,Y,W,\mathcal{W},f_i,\mathcal{G}\}$ is an N-S-N game under DDUs where $X_i,Y$ and $W$ are non-empty compact convex sets. 

(a) If Assumptions \ref{asmp:SetValuedMap} and \ref{asmp:fi} hold, then there exists at least one weak equilibrium point of game $\Theta^{\rm DDU}$. 

(b) If Assumptions \ref{asmp:SetValuedMap} and \ref{asmp:fi} hold with $S=N$, then there exists at least one strong equilibrium point of game $\Theta^{\rm DDU}$.
\end{thm}

Theorem \ref{equilibrium_existence} provides a sufficient condition to the equilibrium existence. Equilibrium existence of game $\Theta^{\rm DDU}$ is established on (i) continuity assumptions; (ii) compactness assumptions; and (iii) (quasi-) convexity assumptions, which is analogous to the existing results for the generalized Nash equilibrium problems \cite{Facchinei2007}. 

The following insights would be helpful: (i) Compactness and continuity assumptions are necessary; (ii) Since leaders try to maximize their payoffs while the uncertainty tries to make them worse off, $f_i(x,y,w)$ is stipulated to be quasi-concave in $x$ and (quasi-) convex in $w$; (iii) Due to the complexity of the leaders-DDUs-followers structure of the game model, existence of the equilibrium imposes more stringent requirements on the bottom-level, which is the follower-level non-cooperative game. $f_i$ has to be (quasi-) linear with respect to followers' strategy $y$. Moreover, followers' reaction map $\mathcal{G}$ is required to have convex graphs; and (iv) No special requirements are imposed on the feature of uncertainty, as long as the values of its feasible strategy set $\mathcal{W}$ are convex and compact.

\subsection{Proof of Theorem \ref{equilibrium_existence}}
\label{sketchproof}
In this subsection, we give a proof of Theorem \ref{equilibrium_existence} based on Kakutani's Fixed-Point Theorem. The main idea is to construct proper set-valued maps according to the definition of the equilibrium and to argue the existence of fixed points. Let Assumptions \ref{asmp:SetValuedMap}-\ref{asmp:fi} hold. We first present three important set-valued maps in the following definition, and then reveal their key properties.

\begin{defn}
Given the game $\Theta^{\rm DDU}$ in Definition \ref{def:model} and Assumption \ref{asmp:fi}, we define the following set-valued maps\footnote{In this definition, $u$, $v$ and $t$ are the variable counterparts of $x$, $y$ and $w$, respectively.}:

(a) For any $i$ in $N$, a set-valued map $\mathcal{H}_i:X_{-i}\times W\rightrightarrows X_i\times Y$ is defined as:
\begin{eqnarray}
\label{def:Hi}
\begin{split}
&\mathcal{H}_i(x_{-i},w)\triangleq\left\{x_i\in X_i,y^{[i]}\in Y:\right.\\
&(x_i,y^{[i]})\in\arg \max_{u_i,v} f_i((u_i,x_{-i}),v,w)\\
&\Big.{\rm s.t.}\ u_i\in \mathcal{X}_i(x_{-i}),v\in \mathcal{G}((u_i,x_{-i}),w)\Big\}.
\end{split}
\end{eqnarray}

(b) $\mathcal{H}_{\rm WPNE}:X\rightrightarrows W\times Y$ is defined on $X$ as
\begin{eqnarray}
\label{def:H0}
\begin{split}
&\mathcal{H}_{\rm WPNE}(x)\triangleq\left\{w\in \mathcal{W}(x),y^{[w]}\in \mathcal{G}(x,w):\right.\\
&\forall t\in\mathcal{W}(x),\forall v\in \mathcal{G}(x,t),\\
&\left.f(x,y^{[w]},w)-f(x,v,t)\notin{\rm int}\mathbb{R}^n_{+}\right\}.
\end{split}
\end{eqnarray}

$\mathcal{H}_{\rm PNE}:X\rightrightarrows W\times Y$ is defined on $X$ as
\begin{eqnarray}
	\label{def:HPNE}
	\begin{split}
		&\mathcal{H}_{\rm PNE}(x)\triangleq\left\{w\in \mathcal{W}(x),y^{[w]}\in \mathcal{G}(x,w):\right.\\
		&\forall t\in\mathcal{W}(x),\forall v\in \mathcal{G}(x,t),\\
		&\left.f(x,y^{[w]},w)-f(x,v,t)\notin\mathbb{R}^n_{+}\backslash\{0\}\right\}.
	\end{split}
\end{eqnarray}

$\mathcal{H}_S:X\rightrightarrows W\times Y$ is defined on $X$ as
\begin{eqnarray}
	\label{def:HS}
	\begin{split}
		\mathcal{H}_S(x)\triangleq\arg &\min_{t,v} \sum_{i\in S}f_i(x,v,t)\\
		&{\rm s.t.}\ t\in\mathcal{W}(x),v\in \mathcal{G}(x,t)
	\end{split}
\end{eqnarray}

(c) A set-valued map
\begin{eqnarray}
\notag
\mathcal{F}_{\rm WPNE}:X\times(\prod_{i=1}^{n+1}Y)\times W\rightrightarrows X\times(\prod_{i=1}^{n+1}Y)\times W
\end{eqnarray}
is defined as:
\begin{eqnarray}
\label{def:F}
\begin{split}
\mathcal{F}_{\rm WPNE}(x,&\mathbf{y},w)\triangleq\\
&\mathcal{H}_{\rm WPNE}(x,y^{[w]})\times\prod\nolimits_{i\in N}\mathcal{H}_i(x_{-i},y^{[i]},w)
\end{split}
\end{eqnarray}
where
\begin{eqnarray}
\notag
\mathbf{y}\triangleq(y^{[1]},...,y^{[n]},y^{[w]})
\end{eqnarray}
is the collection of anticipations on followers' reaction with $y^{[i]}\in Y,\forall i\in N$, and $y^{[w]}\in Y$. Note that in \eqref{def:F}, $\mathcal{H}_i$ and $\mathcal{H}_{\rm WPNE}$ are written as $\mathcal{H}_i(x_{-i},y^{[i]},w)$ and $\mathcal{H}_{\rm WPNE}(x,y^{[w]})$ just for ease of exposition. Actually the value of $\mathcal{H}_i$ does not depend on $y^{[i]}$ and the value of $\mathcal{H}_{\rm WPNE}$ has no relation with $y^{[w]}$, according to \eqref{def:Hi} and \eqref{def:H0}, respectively.

Similarly, a set-valued map
\begin{eqnarray}
	\notag
	\mathcal{F}_{\rm PNE}:X\times(\prod_{i=1}^{n+1}Y)\times W\rightrightarrows X\times(\prod_{i=1}^{n+1}Y)\times W
\end{eqnarray}
is defined as:
\begin{eqnarray}
	\label{def:FPNE}
	\begin{split}
		\mathcal{F}_{\rm PNE}(x,&\mathbf{y},w)\triangleq\\
		&\mathcal{H}_{\rm PNE}(x,y^{[w]})\times\prod\nolimits_{i\in N}\mathcal{H}_i(x_{-i},y^{[i]},w).
	\end{split}
\end{eqnarray}

A set-valued map 
\begin{eqnarray}
	\notag
	\mathcal{F}_S:X\times(\prod_{i=1}^{n+1}Y)\times W\rightrightarrows X\times(\prod_{i=1}^{n+1}Y)\times W
\end{eqnarray}
is defined as
\begin{eqnarray}
	\label{def:FS}
	\mathcal{F}_S(x,\mathbf{y},w)\triangleq\mathcal{H}_S(x,y^{[w]})\times\prod\nolimits_{i\in N}\mathcal{H}_i(x_{-i},y^{[i]},w).
\end{eqnarray}
\end{defn}

\begin{rem}\label{rem:dfn}
\textcolor{white}{xx}

(a) $\mathcal{H}_i(x_{-i},w)$ is the set of situations that satisfy Condition \ref{cond:equilibrium} (c) with $x_{-i}=x_{-i}^*$, $w=w^*$.

(b) $\mathcal{H}_{\rm WPNE}(x)$ is the set of situations that satisfy Condition \ref{cond:equilibrium} (d1) with $x=x^*$.

(c) $\mathcal{H}_{\rm PNE}(x)$ is the set of situations that satisfy Condition \ref{cond:equilibrium} (d2) with $x=x^*$.

(d) It is obvious that $\mathcal{H}_S(x)\subseteq \mathcal{H}_{\rm WPNE}(x)$ and $\mathcal{H}_S(x)|_{S=N}\subseteq \mathcal{H}_{\rm PNE}(x)$ hold for any $x\in X$.

(e) Based on Remark \ref{rem:dfn} (d), obviously, $\mathcal{F}_S(x)\subseteq \mathcal{F}_{\rm WPNE}(x)$ and $\mathcal{F}_S(x)|_{S=N}\subseteq \mathcal{F}_{\rm PNE}(x)$ hold for any $x\in X$. Thus a fixed point of $\mathcal{F}_S$ is also a fixed point of $\mathcal{F}_{\rm WPNE}$. A fixed point of $\mathcal{F}_S|_{S=N}$ is also a fixed point of $\mathcal{F}_{\rm PNE}$.
\end{rem}

The following lemma shows the relationship between the fixed point of set-valued map $\mathcal{F}_{\rm WPNE}$ and the weak equilibrium point of $\Theta^{\rm DDU}$, and the relationship between the fixed point of set-valued map $\mathcal{F}_{\rm PNE}$ and the strong equilibrium point of $\Theta^{\rm DDU}$. 
\begin{lem}
	\label{relationship}
	\textcolor{white}{xx}
	
	(a) Let $(x^*,w^*,\mathbf{y}^*)$ be a fixed point of $\mathcal{F}_{\rm WPNE}$, then $(x^*,w^*,\mathbf{y}^*)$ is a weak equilibrium of game $\Theta^{\rm DDU}$.
	
	(b) Let $(x^*,w^*,\mathbf{y}^*)$ be a fixed point of $\mathcal{F}_{\rm PNE}$, then $(x^*,w^*,\mathbf{y}^*)$ is a strong equilibrium of game $\Theta^{\rm DDU}$.
\end{lem}
Lemma \ref{relationship} can be proved by finding out that the fixed points of $\mathcal{F}_{\rm WPNE}$ satisfy Condition \ref{cond:equilibrium} (a)-(c) and (d1); and the fixed points of $\mathcal{F}_{\rm PNE}$ satisfy Condition \ref{cond:equilibrium} (a)-(c) and (d2). Lemma \ref{relationship} and Remark \ref{rem:dfn} (d) together transform the existence of the weak equilibrium of $\Theta^{\rm DDU}$ into the existence of a fixed point of $\mathcal{F}_{\rm S}$; and transform the existence of the strong equilibrium of $\Theta^{\rm DDU}$ into the existence of a fixed point of $\mathcal{F}_{ S}|_{S=N}$. Next, we give some preliminary results on $\mathcal{H}_i$, $\mathcal{H}_{S}$, and $\mathcal{F}_{S}$.

\begin{lem}
\label{lemma:Hi}
Let Assumptions \ref{asmp:SetValuedMap}-\ref{asmp:fi} hold.
For any $i\in N$, the set-valued map $\mathcal{H}_i$ defined in \eqref{def:Hi} has the following properties:

(a) $\mathcal{H}_i$ is non-empty in $X_{-i}\times W$;

(b) For any $x_{-i}\in X_{-i}$ and $w\in W$, $\mathcal{H}_i(x_{-i},w)$ is a compact set;

(c) For any $x_{-i}\in X_{-i}$ and $w\in W$, $\mathcal{H}_i(x_{-i},w)$ is a convex set;

(d) $\mathcal{H}_i(x_{-i},w)$ is upper semi-continuous with respect to $x_{-i}$ and $w$.
\end{lem}

\begin{lem}
	\label{lemma:H0}
	Let Assumptions \ref{asmp:SetValuedMap}-\ref{asmp:fi} hold.
	The set-valued map $\mathcal{H}_S$ defined in \eqref{def:HS} has the following properties:

	(a) $\mathcal{H}_S$ is non-empty in $X$;
	
	(b) For any $x\in X$, $\mathcal{H}_S(x)$ is a convex set;
	
	(c) For any $x\in X$, $\mathcal{H}_S(x)$ is a compact set;
	
	(d) $\mathcal{H}_S(x)$ is upper semi-continuous with respect to $x$. 
\end{lem}

\begin{lem}
	\label{lemma:F}
	Let Assumptions \ref{asmp:SetValuedMap}-\ref{asmp:fi} hold.
	The set-valued map $\mathcal{F}_S$ defined in \eqref{def:FS} has the following properties:

	(a) For any $(x,\mathbf{y},w)\in X\times(\prod_{i=1}^{n+1}Y)\times W$, $\mathcal{F}_S(x,\mathbf{y},w)$ is non-empty, compact and convex;
		
	(b) $\mathcal{F}_S(x,\mathbf{y},w)$ is upper semi-continuous with respect to $(x,\mathbf{y},w)$.
\end{lem}

Key properties of $\mathcal{H}_i$, $\mathcal{H}_S$ and $\mathcal{F}_S$ are given by Lemma \ref{lemma:Hi}, Lemma \ref{lemma:H0} and Lemma \ref{lemma:F}, respectively.
Proof of Lemma \ref{lemma:Hi} and \ref{lemma:H0} can be found in the Appendix \ref{appendix2}.
They indicate that $\mathcal{H}_i$ and $\mathcal{H}_S$ are upper semi-continuous set-valued maps with non-empty compact convex values. Lemma \ref{lemma:F} can be induced from Lemma \ref{lemma:Hi} and \ref{lemma:H0} as follows: Since the values of $\mathcal{H}_i$ and $\mathcal{H}_S$ are non-empty, compact, and convex, so is the value of $\mathcal{F}_S$, by recalling that the value of $\mathcal{F}_S$ is the Cartesian product of the value of $\mathcal{H}_i$ and $\mathcal{H}_S$. Since $\mathcal{H}_i$ and $\mathcal{H}_S$ are upper semi-continuous according to Lemma \ref{lemma:Hi} and \ref{lemma:H0}, so is $\mathcal{F}_S$, by noting the continuity of composite set-valued map \cite[Chapter 1.1 Proposition 1]{Aubin:Differential}.

Now we are ready to give the proof of Theorem \ref{equilibrium_existence}.

\begin{proof}[Proof of Theorem \ref{equilibrium_existence}]
\textcolor{white}{text}

The existence of a fixed point of $\mathcal{F}_S$ can be guaranteed by invoking Kakutani's fixed point theorem (Lemma \ref{kakutani}), since $\mathcal{F}_S$ is upper semi-continuous with non-empty compact convex values, according to Lemma \ref{lemma:F}. Specifically, there exists $(x^*,\mathbf{y}^{*},w^*)\in X\times (\prod_{i= 1}^{n+1}Y)\times W$ such that
\begin{eqnarray} 
\notag
(x^*,\mathbf{y}^{*},w^*)\in \mathcal{F}_S(x^*,\mathbf{y}^{*},w^*).
\end{eqnarray}
According to the definition of the equilibrium of $\Theta^{\rm DDU}$ in Definition \ref{def:eq}, the fixed point of $\mathcal{F}_S$ is exactly a weak equilibrium point of game $\Theta^{\rm DDU}$. If Assumption \ref{asmp:fi} (b) is satisfied with $S=N$, the fixed point of $\mathcal{F}_S|_{S=N}$ is exactly a strong equilibrium point of game $\Theta^{\rm DDU}$. Proof of Theorem \ref{equilibrium_existence} is completed. 
\end{proof}

Theorem \ref{equilibrium_existence} provides a sufficient condition to the existence of the weak and strong equilibrium of $\Theta^{\rm DDU}$, namely, Assumptions \ref{asmp:SetValuedMap}-\ref{asmp:fi}. If we only focus on the weak equilibrium points of $\Theta^{\rm DDU}$, this sufficient condition can be further relaxed by substituting Assumption \ref{asmp:fi} (b) with the following hypothesis.
\begin{assum}
\label{asmp:fi-extend}
There exists a non-empty subset of $N$ which is $S\subseteq N,S\neq \emptyset$ such that for any $x\in X$ and $i\in S$, $-f_i(x,y,w)$ is quasi-concave with respect to $(y,w)$.
\end{assum}
Assumption \ref{asmp:fi-extend} differs from Assumption \ref{asmp:fi} (b) in substituting the concavity of $-f_i$ with quasi-concavity of $-f_i$, which is a more relaxed condition than the original one. Results based on Assumption \ref{asmp:fi-extend} are stated in the following corollary.

\begin{cor}
\label{cor}
Game $\Theta^{\rm DDU}$ is an N-S-N game under DDUs where $X_i,Y$ and $W$ are non-empty compact convex sets. Assume that Assumption \ref{asmp:SetValuedMap}, Assumption \ref{asmp:fi} (a) and (c), and Assumption \ref{asmp:fi-extend} hold, then there exists at least one weak equilibrium point of game $\Theta^{\rm DDU}$. 
\end{cor}

Corollary \ref{cor} can be proved likewise by properly constructing set-valued maps. Specifically, choose any $\hat{i}$ from $S$ and consider the following set-valued maps
\begin{eqnarray}
\notag
\begin{split}
\mathcal{H}_{\hat{i}\in S}(x)\triangleq \arg &\min_{t,v} f_{\hat{i}}(x,y,w)\\
&{\rm s.t.}\ t\in\mathcal{W}(x),v\in \mathcal{G}(x,t)
\end{split}
\end{eqnarray}
\begin{eqnarray}
\notag
\mathcal{F}_{\hat{i}\in S}(x,\mathbf{y},w)\triangleq\mathcal{H}_{\hat{i}\in S}(x,y^{[w]})\times\prod\nolimits_{i\in N}\mathcal{H}_i(x_{-i},y^{[i]},w).
\end{eqnarray}
which are upper semi-continuous and with non-empty compact convex values.

\section{An Illustrative Example}
\label{mathproblem}
Consider a game $\Theta^{\rm DDU}$ with two leaders $N=\{1,2\}$, two followers $M=\{1,2\}$, and one decision-dependent uncertain parameter. Strategies of the leaders are denoted by $x_i\in X_i\subseteq\mathbb{R}^{2},i\in N$. Strategies of the followers are $y_j\in \mathbb{R}^1,j\in M$. The uncertain parameter is $w\in W\subseteq \mathbb{R}^1$. $X_1=X_2=[0,1]^2,W=[-4,4]$. The payoff that leader $i$ would like to maximize is
\begin{eqnarray}
\notag
&f_i(x,y,w)=a_i^{\mathsf{T}}x_i+b_i^{\mathsf T}y+c_i(d_i-w)^2,&
\end{eqnarray}
where $a_i,b_i\in\mathbb{R}^2,c_i,d_i\in\mathbb{R}^1$ are constant parameters and $y=(y_1,y_2)^{\mathsf{T}}$ is the collection of the followers' strategies. Specific values of $a_1=a_2=(1.3,0)^{\mathsf{T}},b_1=(-1.2,-1.2)^{\mathsf{T}},b_2=(0.4,0.4)^{\mathsf{T}},c_1=c_2=0.2,d_1=d_2=2$ are considered in this case. The uncertain $w$ is decision-dependent, the feasible map of which is
\begin{eqnarray*}
&\mathcal{W}(x)=\left\{w\in W|{w}^{\rm min}(x)\le w\le {w}^{\rm max}(x)\right\}&
\end{eqnarray*}
where
\begin{eqnarray*}
	\begin{split}
		&{w}^{\rm min}(x)=-4+\sigma_1^{\mathsf{T}}x_1+\sigma_2^{\mathsf{T}}x_2\\
		&{w}^{\rm max}(x)=4-\sigma_1^{\mathsf{T}}x_1-\sigma_2^{\mathsf{T}}x_2
	\end{split}
\end{eqnarray*}
and $\sigma_1=\sigma_2=(0,2)^{\mathsf{T}}$. By selecting appropriate $x$, the leaders can condense the range that the uncertainty parameter $w$ varies within. The followers' problems are
\begin{eqnarray*}
\begin{split}
y_1=&\max_{v\in\mathbb{R}^2}\ e_1^{\mathsf{T}}v\\
&{\rm s.t.}\ v_1\ge 0,g_1^{\mathsf{T}}x_1+k_1^{\mathsf{T}}x_2+h_1^{\mathsf{T}}v=w+\alpha_1 y_2
\end{split}\\
\begin{split}
y_2=&\max_{v\in\mathbb{R}^2}\ e_2^{\mathsf{T}}v\\
&{\rm s.t.}\ v_1\ge 0,g_2^{\mathsf{T}}x_1+k_2^{\mathsf{T}}x_2+h_2^{\mathsf{T}}v=w+\alpha_2 y_1
\end{split}
\end{eqnarray*}
where $e_1=e_2=(1,2)^{\mathsf{T}},g_1=g_2=(1,0)^{\mathsf{T}},k_1=k_2=(1,0)^{\mathsf{T}},h_1=h_2=(2,1)^{\mathsf{T}}\in\mathbb{R}^2,\alpha_1=\alpha_2=1$ are constant coefficients.

Next, we show the above case satisfies Assumptions \ref{asmp:SetValuedMap}-\ref{asmp:fi}. It is easy to verify that $f_i$ is continuous; $f_i$ is concave with respect to $x$ and $y$, and is convex with respect to $w$ and $y$; According to Lemma \ref{lemma:marginal}, the followers' response map $\mathcal{G}$ is continuous and the graph of $\mathcal{G}$ is convex by noting the explicit formation as follows
\begin{eqnarray*}
\mathcal{G}(x,w)=\left\{y\in\mathbb{R}^2:
\begin{array}{l}
y_1=-2(w-x_{1,1}-x_{2,1})\\
y_2=-2(w-x_{1,1}-x_{2,1})
\end{array}
\right\}
\end{eqnarray*}
where $x_{1,1}$ is the first element of $x_1$ and $x_{2,1}$ is the first element of $x_2$. Also, it is clear that $\mathcal{W}(x)$ is continuous and has compact convex values. Thus there exists an equilibrium of the game according to Theorem \ref{equilibrium_existence}.

To solve the strong equilibrium of the game, we substitute the followers' problems with the explicit form $\mathcal{G}(x,w)$ and apply the best response (BR) of BR algorithm (with Jacobi iteration) by iteratively solving the sub-problems of the two leaders till an equilibrium point is achieved. As for the sub-problem of the virtual player $w$, a pre-set weight factor $\lambda\in[0,1]$ is assigned to the multi-objective $f=(f_1,f_2)^{\mathsf{T}}$ to derive a specific Pareto solution. Thus the virtual player minimizes the weighted sum of leaders' payoffs $\lambda f_1+(1-\lambda)f_2$. Let $(x_1^*,x_2^*,w^*,y^*)$ denote the equilibrium of the game and $f_i^*=f_i(x^*,y^*,w^*),i\in N$. The following figures present the equilibrium results. 

\begin{figure}[!ht]
	\centering
	\includegraphics[width=0.45\textwidth]{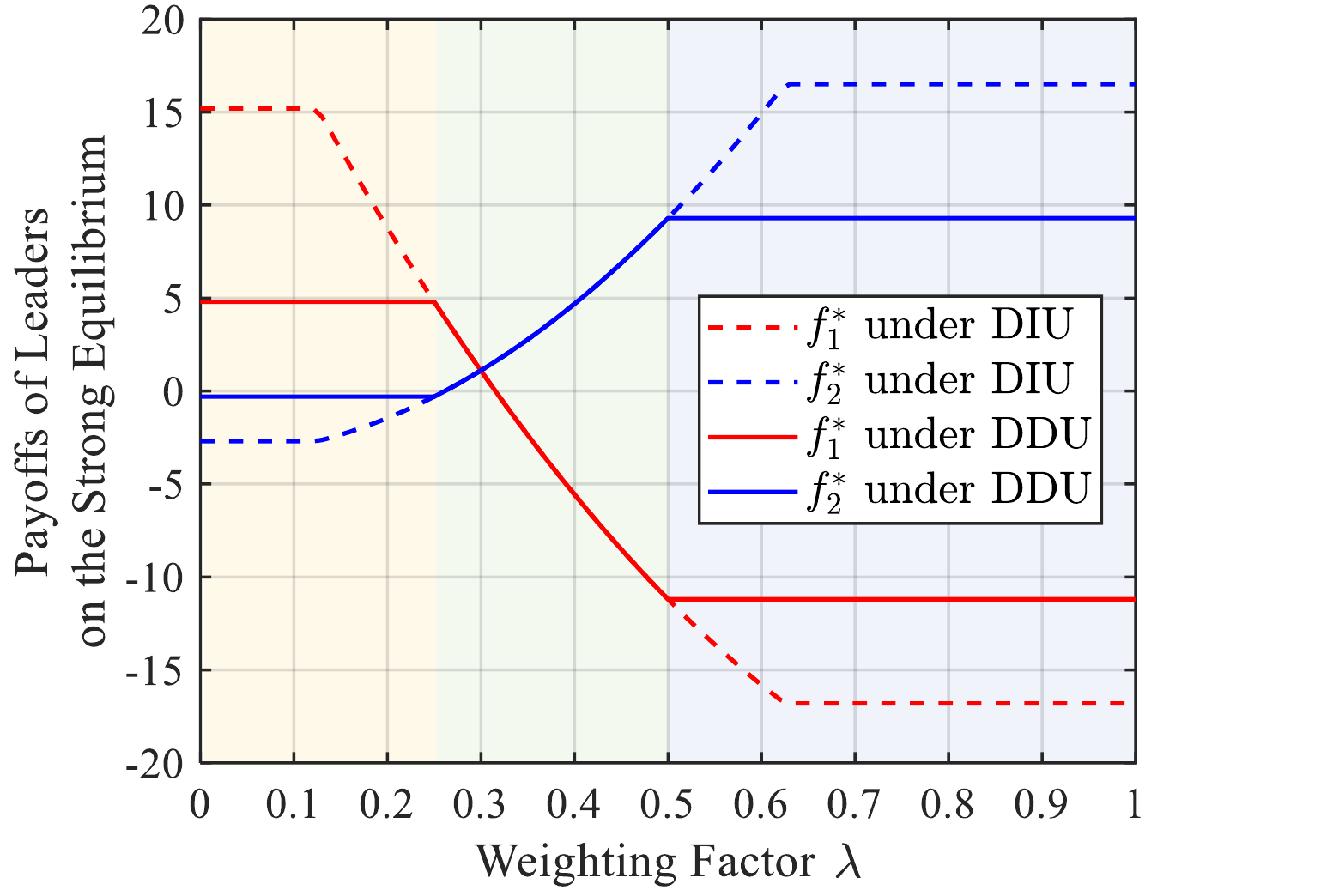}
	\caption{Payoffs of the two leaders under DIU and DDU.}
	\label{obj}
\end{figure}
Figure \ref{obj} shows how the DDU affects the strong equilibrium of the N-S-N game. For comparison, equilibrium under DIU is considered by setting $\mathcal{W}(x)=W$. The equilibrium under DIU is $x_1^*=(0,0)^{\mathsf{T}},x_2^*=(1,0)^{\mathsf{T}}$ for all $\lambda\in [0,1]$. As for the equilibrium under DDU, when $\lambda$ falling into $[0,0.25]$ (the yellow region in Figure \ref{obj}), the equilibrium under DDU is $x_1^*=(0,0)^{\mathsf{T}},x_2^*=(1,1)^{\mathsf{T}}$, indicating that the leader 2 promotes his payoff by restricting the uncertainty set of $w$ to $[-2,2]$. Though leader 1 has no incentive to do so, his payoff in this non-cooperative game is reduced due to the choice of leader 2 and the resulting worst-case $w^*$. When $\lambda$ falling into $[0.5,1]$ (the blue region in Figure \ref{obj}), the equilibrium under DDU is  $x_1^*=(0,1)^{\mathsf{T}},x_2^*=(1,0)^{\mathsf{T}}$, indicating the leader 1 would like to condense the uncertainty to improve his payoff, whereas leader 2 would suffer a loss. If $0.25<\lambda<0.5$ (the green region in Figure \ref{obj}), the equilibrium under DDU becomes $x_1^*=(0,0)^{\mathsf{T}},x_2^*=(1,0)^{\mathsf{T}}$, indicating that both the two leaders have no incentives to derive a better uncertainty set. Thus when $0.25<\lambda<0.5$, the equilibrium and the corresponding payoff of leaders under DDU and DIU are the same.
\vspace{0.5cm} 
\begin{figure}[!ht]
	\centering
	\includegraphics[width=0.45\textwidth]{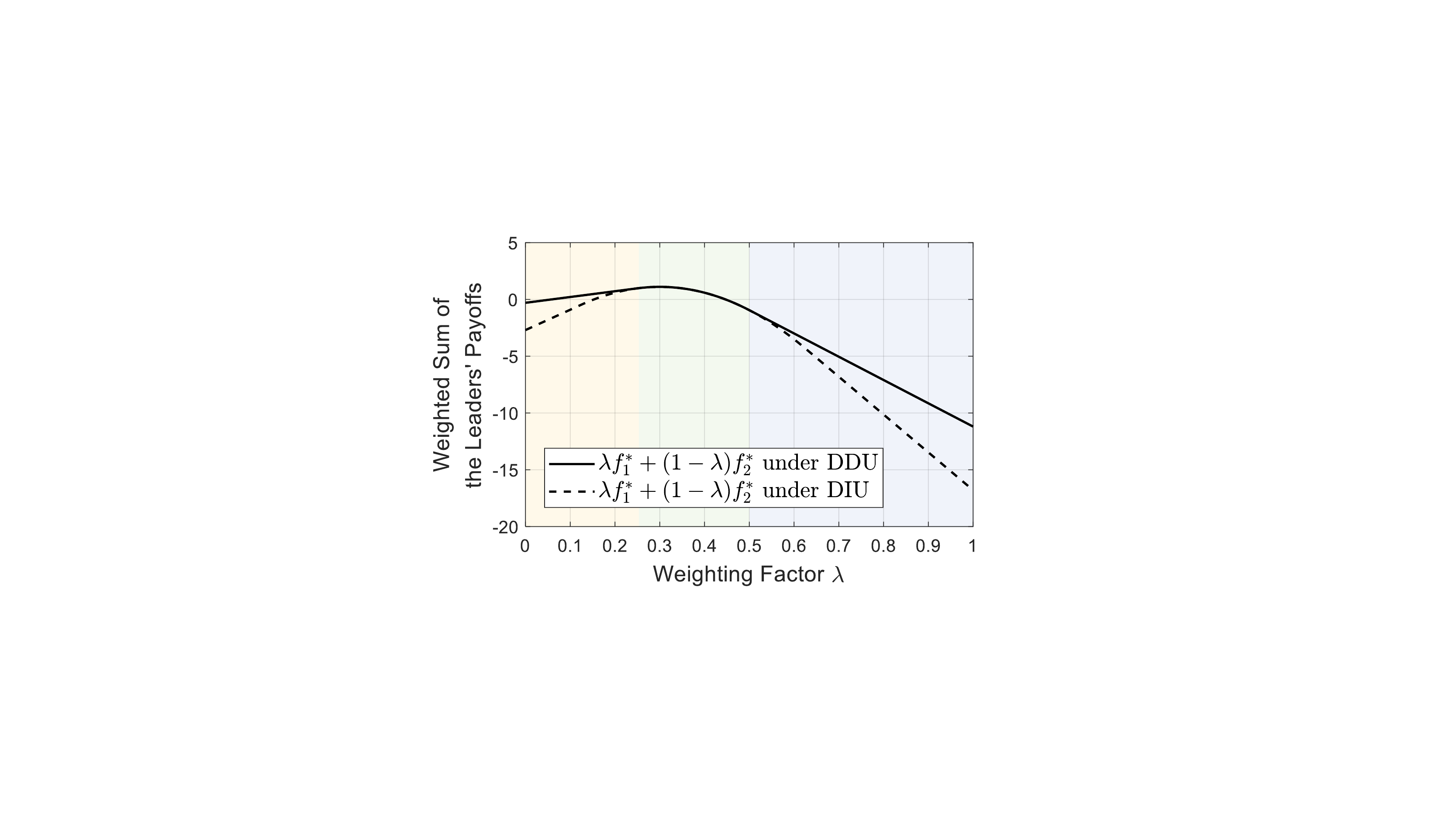}
	\caption{Weighted sum of leaders' payoffs under DIU and DDU.}
	\label{pareto1}
\end{figure}

\begin{figure}[!ht]
	\centering
	\includegraphics[width=0.45\textwidth]{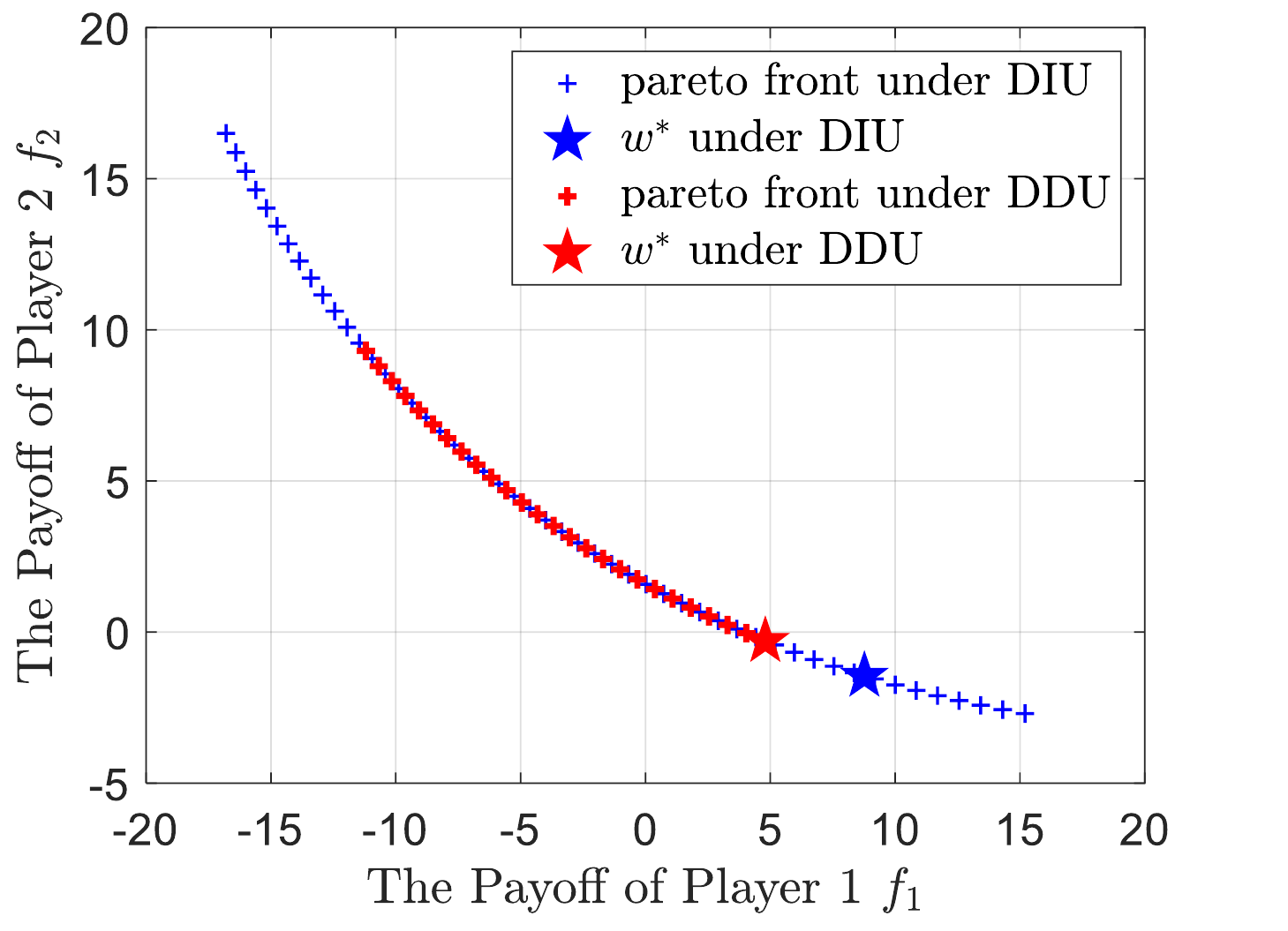}
	\caption{Pareto optimality at $x^*$ ($\lambda=0.2$) under DIU and DDU.}
	\label{pareto2}
\end{figure}

Figure \ref{pareto1} presents how the union utility of leaders is improved by taking the consideration of DDU. The virtual player $w$ tries to worse off the payoff of the two leaders, minimizing the weighted sum $\lambda f_1 + (1-\lambda)f_2$. Thus $\lambda f_1 + (1-\lambda)f_2$ can be viewed as the united utility of the two leaders. Figure \ref{pareto1} shows that the equilibrium $\lambda f_1^* + (1-\lambda)f_2^*$ under DDU is always greater than or equal to that under DIU, indicating that an improvement on leaders' union utility can be derived by exercising proactive control on the uncertainty.

Next, we fix the weight factor $\lambda$ as 0.2. When $\lambda=0.2$, the equilibrium under DDU is $x_1^*=(0,0)^{\mathsf{T}},x_2^*=(1,1)^{\mathsf{T}},w^*=2,y_1^*=y_2^*=-2$ while the equilibrium under DIU is $x_1^*=(0,0)^{\mathsf{T}},x_2^*=(1,0)^{\mathsf{T}},w^*=2.8,y_1^*=y_2^*=-3.6$. Figure \ref{pareto2} depicts the Pareto front of the virtual player's multi-objective optimization on the equilibrium point, i.e., the Pareto front of problem \eqref{def:moo}. The pentagrams denote the Pareto solution that corresponds with $\lambda=0.2$. It is observed that the Pareto front under DDU becomes shorter, indicating a more limited varying range of the uncertainty. The worst-case $w^*$ that minimizes $\lambda f_1 + (1-\lambda)f_2$ also changes with the presence of DDU.

\section{Conclusion}
\label{discussion}
This paper explores the existence of the equilibrium of Nash-Stackelberg-Nash games under decision-dependent uncertainties. We have mathematically formulated the class of games and rigorously define the Nash equilibrium. We have proved the existence of the Nash equilibrium. It is revealed that the players can leverage the dependency of uncertainties on decisions to restrict the negative influence of the uncertainties on their payoffs, which provides an insight in playing this class of games. The conducted work in this paper can be viewed as an extension of Nash-Stackelberg-Nash games to incorporate decision-dependent uncertainties with the idea of robust optimization. The model naturally encompasses the variants including the conventional Stackelberg game, the single-leader multi-follower game, and the multi-leader single-follower game. The model also enables potential extensions to other forms of game problems such as multi-cluster games.


\section{Appendix: Proof of Lemma \ref{lemma:Hi} and Lemma \ref{lemma:H0}}
\label{appendix2}
We start the proof of Lemma \ref{lemma:Hi} with the following lemma.

\begin{lem}
	\label{lemmaK}
	$\forall i\in N$, if $\mathcal{X}_i:X_{-i}\rightrightarrows X_i$ satisfies Assumption \ref{asmp:SetValuedMap} (a) and $\mathcal{G}:X\times W\rightrightarrows Y$ satisfies Assumption \ref{asmp:SetValuedMap} (b), then the set-valued map below
	\begin{eqnarray}
	\label{def:Ki}
	\begin{split} &\mathcal{K}_i(x_{-i},w)\triangleq\Big\{x_i\in X_i,y\in Y:\Big.\\
	&\quad\quad\quad\Big.x_i\in \mathcal{X}_i(x_{-i}),y\in \mathcal{G}((x_i,x_{-i}),w)
	\Big\}
	\end{split}
    \end{eqnarray}
    has the following properties:

	(a) $\mathcal{K}_i:X_{-i}\times W\rightrightarrows X_i\times Y$ is continuous in $X_{-i}\times W$;
	
	(b) $\forall x_{-i}\in X_{-i},\forall w\in W$, $\mathcal{K}_i(x_{-i},w)$ is non-empty compact convex set.
\end{lem}
\begin{proof}[Proof of Lemma \ref{lemmaK}]
	\textcolor{white}{text}
	
	\textit{\textbf{Assertion (a):}} We prove $\mathcal{K}_i$ is lower semi-continuous on $X_{-i}\times W$ by recalling the fact that $\mathcal{K}_i$ is lower semi-continuous if and only if $\forall x_{-i}^k\rightarrow x_{-i}$, $ w^k\rightarrow w$ and $(x_{i},y)\in \mathcal{K}_{i}(x_{-i},w)$, there exists $(x_i^k,y^k)\in\mathcal{K}_i(x_{-i}^k,w^k)$ such that $x_i^k\rightarrow x_i,y^k\rightarrow y$ \cite[Chapter 1.1 Definition 2]{Aubin:Differential}. Since $\mathcal{X}_i$ is lower semi-continuous, there exists $x_i^k\in \mathcal{X}_i(x_{-i}^k)$ such that $x_i^k\rightarrow x_i$. Similarly, since $\mathcal{G}$ is lower semi-continuous, there exists $y^k\in \mathcal{G}(x^k,w^k)$ such that $y^k\rightarrow y$. Thus $\forall x_{-i}^k\rightarrow x_{-i}$, $w^k\rightarrow w$, $\forall (x_i,y)\in \mathcal{K}_{i}(x_{-i},w)$, we find sequence $(x_i^k,y^k)\in\mathcal{K}_i(x_{-i}^k,w^k)$ such that $x_i^k\rightarrow x_i,y^k\rightarrow y$.
	
	Next, we would like to show that the graph of $\mathcal{K}_i$ is closed. Since both $\mathcal{X}_i$ and $\mathcal{G}$ are continuous set-valued map with compact domains and compact values, the graph of $\mathcal{X}_i$ and $\mathcal{G}$ are closed \cite[Chapter 1.4.1 Proposition 1.4.8]{Aubin:Setvalued}. Thus $\forall \{x_{-i}^k\}\rightarrow x_{-i}$, $\{w^k\}\rightarrow w$, $\{x_i^k\in\mathcal{X}_i(x_{-i}^k)\}\rightarrow x_i$, $\{y^k\in\mathcal{G}(x^k,w^k)\}\rightarrow y$, there are $x_i\in \mathcal{X}_i(x_{-i})$ and $y\in \mathcal{G}(x,w)$, which means $(x_i,y)\in \mathcal{K}_i(x_{-i},w)$. So the graph of $\mathcal{K}_i$ is closed according to definition of closed set-valued map. Thus $\mathcal{K}_i$ is upper semi-continuous \cite[Chapter 1.1, Corollary 1]{Aubin:Differential}. Since $\mathcal{K}_i$ is both lower and upper semi-continuous, $\mathcal{K}_i$ is continuous set-valued map.
	
	\textit{\textbf{Assertion (b):}} Since $\mathcal{X}_i$ and $\mathcal{G}$ are non-empty, so is $\mathcal{K}_i$. Moreover, since $\mathcal{K}_i$ is defined on compact domain and $\mathcal{K}_i$ is closed, it is compact. Next, we prove that $\mathcal{K}_i(x_{-i},w)$ is convex set for any $x_{-i}\in X_{-i},w\in W$. For any $(x_i^1,y^1),(x_i^2,y^2)\in\mathcal{K}_i(x_{-i},w)$ and $\forall \gamma\in(0,1)$, there are
	\begin{eqnarray}
	\notag
	y^1\in \mathcal{G}((x_i^1,x_{-i}),w^1),y^2\in \mathcal{G}((x_i^2,x_{-i}),w^2).
	\end{eqnarray}
	Denote $\hat{x}=((\gamma x_i^1+(1-\gamma)x_i^2),x_{-i})$. Since $\mathcal{G}$ is convex set-valued map (i.e., the graph of $\mathcal{G}$ is convex),
	\begin{eqnarray}
	\notag
	\gamma\mathcal{G}((x_i^1,x_{-i}),w^1)+(1-\gamma)\mathcal{G}((x_i^2,x_{-i}),w^2)\\
	\notag
	\subseteq \mathcal{G}(\hat{x},\gamma w^1+(1-\gamma)w^2).
	\end{eqnarray}
	Thus
	\begin{eqnarray}
	\label{appen_3}
	\gamma y^1+(1-\gamma)y^2\in 
	\mathcal{G}(\hat{x},\gamma w^1+(1-\gamma)w^2).
	\end{eqnarray}
	Since $\mathcal{X}_i(x_{-i})$ is convex,
	\begin{eqnarray}
	\label{appen_4}
	\gamma x_i^1+(1-\gamma)x_i^2\in\mathcal{X}_i(x_{-i}). 
	\end{eqnarray}
	Thus \eqref{appen_3} and \eqref{appen_4} together implies that
	\begin{eqnarray}
	\notag
	(\gamma x_i^1+(1-\gamma)x_i^2,\gamma y^1+(1-\gamma)y^2)\in \mathcal{K}_i(x_{-i},w),
	\end{eqnarray}
	which completes the proof.

\end{proof}

Next, we give the proof of Lemma \ref{lemma:Hi} (a)-(d) by sequence.
\begin{proof}[Proof of Lemma \ref{lemma:Hi}]
\textcolor{white}{text}

\textit{\textbf{Assertion (a):}} Since Assumption \ref{asmp:SetValuedMap} holds,
	$\mathcal{X}_i(x_{-i})$ is non-empty compact set for any $x_{-i}\in X_{-i}$, and so is  $\mathcal{G}(x,w)$ for any $x\in X$ and $w\in W$. Thus $\forall i\in N,\forall x_{-i}\in X_{-i}$ and $\forall w\in W$, $\mathcal{H}_i(x_{-i},w)$ is non-empty, completing the proof of assertion (a).

	\textit{\textbf{Assertion (b):}} Next we  show that $\mathcal{H}_i(x_{-i},w)$ is compact for any $(x_{-i},w)\in X_{-i}\times W$. Let $f_i^*$ denote the optimal objective value of the maximization problem in \eqref{def:Hi}. Then $\mathcal{H}_i(x_{-i},w)$ can be represented as the intersection of two closed sets:
	\begin{eqnarray}
	\notag
	\mathcal{H}_i(x_{-i},w)=\mathcal{K}_i(x_{-i},w)\cap \left\{(x_i,y)|f_i(x,y,w)\ge f_i^*\right\}.
	\end{eqnarray}
	where $\mathcal{K}_i$ is defined in \eqref{def:Ki} and is a compact subset of $X\times Y$ as stated in Lemma \ref{lemmaK} (b). $\left\{(x_i,y)|f_i(x,y,w)\ge f_i^*\right\}$ is closed since $f_i$ is continuous as assumed in Assumption \ref{asmp:fi} (c). Thus $\mathcal{H}_i(x_{-i},w)$ must be a closed subset of $X_i\times Y$. Moreover, since $X_i$ and $Y$ are compact, $\mathcal{H}_i(x_{-i},w)$ is compact, which completes the proof of assertion (b).

	\textit{\textbf{Assertion (c):}} $\forall (x_i^1,y^1),(x_i^2,y^2)\in \mathcal{H}_i(x_{-i},w)$ and  $\forall\gamma\in(0,1)$, denote
	\begin{eqnarray}
	\notag
	\hat{x}_{i}=\gamma x_i^1+(1-\gamma)x_i^2,
	\end{eqnarray}
    and
    \begin{eqnarray}
    \notag
    \hat{y}=\gamma y^1+(1-\gamma)y^2.
    \end{eqnarray}
	since $\mathcal{X}_i(x_{-i})$ is a convex set as assumed in Assumption \ref{asmp:SetValuedMap} (a), $\hat{x}_{i}\in \mathcal{X}_i(x_{-i})$. Since $\mathcal{G}$ is convex set-valued map as assumed in Assumption \ref{asmp:SetValuedMap} (b), 
	\begin{eqnarray}
	\notag
	\begin{split}
	\hat{y}\ =\ &\gamma y^1+(1-\gamma)y^2\\
	\in\ &\gamma \mathcal{G}((x_i^1,x_{-i}),w)+(1-\gamma)\mathcal{G}((x_i^2,x_{-i}),w) \\
	\subseteq\ &\mathcal{G}((\gamma x_i^1+(1-\gamma)x_i^2,x_{-i}),w).
	\end{split}
	\end{eqnarray}
	Since $f_i((x_i,x_{-i}),y,w)$ is quasi-concave with respect to $(x_i,y)$ as assumed in Assumption \ref{asmp:fi} (a), 
	\begin{eqnarray}
	\notag
    \begin{split}
	&f_i((\hat{x}_i,x_{-i}),\hat{y},w)
	\ge\\ &\quad\quad\quad\min\left\{f_i((x_i^1,x_{-i}),y^1,w),f_i((x_i^2,x_{-i}),y^2,w)\right\}.
	\end{split}
	\end{eqnarray}
	Since $(x_i^1,y^1),(x_i^2,y^2)\in \mathcal{H}_i(x_{-i},w)$,
	\begin{eqnarray}
	\notag
	\begin{split}
	f_i((x_i^1,x_{-i}),y^1,w)&=\max_{u_i,v}f_i((u_i,x_{-i}),v,w)\\
	&{\rm s.t.}\quad u_i\in \mathcal{X}_i(x_{-i}),v\in \mathcal{G}((u_i,x_{-i}),w),\\
	f_i((x_i^2,x_{-i}),y^2,w)&=\max_{u_i,v}f_i((u_i,x_{-i}),v,w)\\
	&{\rm s.t.}\quad u_i\in \mathcal{X}_i(x_{-i}),v\in \mathcal{G}((u_i,x_{-i}),w).
	\end{split}
	\end{eqnarray}
	Thus 
	\begin{eqnarray}
	\notag
	\begin{split}
	f_i((\hat{x}_i,x_{-i}),\hat{y},w)&\ge \max_{u_i,v}\ f_i((u_i,x_{-i}),v,w)\\
	&\ {\rm s.t.}\ u_i\in \mathcal{X}_i(x_{-i}),v\in \mathcal{G}((u_i,x_{-i}),w),
	\end{split}
	\end{eqnarray}
	namely, 
	\begin{eqnarray}
	\notag
	\gamma(x_i^1,y^1)+(1-\gamma)(x_i^2,y^2)\in \mathcal{H}_i(x_{-i},w).
	\end{eqnarray}
	Thus $\mathcal{H}_i(x_{-i},w)$ is a convex set, completing the proof of assertion (c).

	\textit{\textbf{Assertion (d):}} Recall the definition of $\mathcal{K}_i$ in Lemma \ref{lemmaK}, $\mathcal{H}_i(x_{-i},w)$ can be rewritten as
	\begin{eqnarray}
	\notag
	\begin{split}
	\mathcal{H}_i(x_{-i},w)=\arg&\max_{u_i,v} f_i((u_i,x_{-i}),v,w)\\
	&\ {\rm s.t.}\ (u_i,v)\in \mathcal{K}_i(x_{-i},w).
	\end{split}
	\end{eqnarray}
	Since $f_i$ is continuous and $\mathcal{K}_i$ is continuous set-valued map with compact values, the marginal map $\mathcal{H}_i$ is upper semi-continuous according to Lemma \ref{lemma:marginal}. This completes the proof of assertion (d).
\end{proof}

We start the proof of Lemma \ref{lemma:H0} with the following lemma.
\begin{lem}
	\label{lemmaJ}
	If $\mathcal{G}:X\times W\rightrightarrows Y$ and $\mathcal{W}:X\rightrightarrows W$ satisfies Assumption \ref{asmp:SetValuedMap} (b) and (c), respectively, then the set-valued map below
	\begin{eqnarray}
	\label{def:J}
	\mathcal{J}(x)\triangleq\left\{(w,y):
	w\in \mathcal{W}(x),y\in \mathcal{G}(x,w)
	\right\}
	\end{eqnarray}
	has the following properties:
	
	(a) $\mathcal{J}:X\rightrightarrows W\times Y$ is continuous in $X$;
	
	(b) $\forall x\in X$, $\mathcal{J}(x)$ is non-empty compact convex set.
\end{lem}
\begin{proof}[Proof of Lemma \ref{lemmaJ}]
	\textcolor{white}{text}
		
	\textit{\textbf{Assertion (a):}} Similar to the proof of Lemma \ref{lemmaK}, we prove that $\mathcal{J}$ is lower semi-continuous on $X$ by recalling the fact that $\mathcal{J}$ is lower semi-continuous if and only if $\forall x^k\rightarrow x$ and $\forall (w,y)\in \mathcal{J}(x)$, there exists $(w^k,y^k)\in\mathcal{J}(x^k)$ such that $w^k\rightarrow w,y^k\rightarrow y$\cite[Chapter 1.1 Definition 2]{Aubin:Differential}. Since $\mathcal{W}$ is lower semi-continuous, there exists $w^k\in \mathcal{W}(x^k)$ such that $w^k\rightarrow w$. Similarly, since $\mathcal{G}$ is lower semi-continuous, there exists $y^k\in \mathcal{G}(x^k,w^k)$ such that $y^k\rightarrow y$. Thus $\forall x^k\rightarrow x$, $\forall (w,y)\in \mathcal{J}(x)$, we find sequence $(w^k,y^k)\in\mathcal{J}(x^k)$ such that $w^k\rightarrow w,y^k\rightarrow y$.
	
	Next, we would like to show that the graph of $\mathcal{J}$ is closed. Since both $\mathcal{W}$ and $\mathcal{G}$ are continuous set-valued map with compact domains and compact values, the graph of $\mathcal{W}$ and $\mathcal{G}$ are closed\cite[Chapter 1.4.1 Proposition 1.4.8]{Aubin:Setvalued}. Thus $\forall \{x^k\}\rightarrow x$, $\{w^k\in\mathcal{W}(x^k)\}\rightarrow w$, $\{y^k\in\mathcal{G}(x^k,w^k)\}\rightarrow y$, there are $w\in \mathcal{W}(x)$ and $y\in \mathcal{G}(x,w)$. So the graph of $\mathcal{J}$ is closed according to definition of closed set-valued map. Thus $\mathcal{J}$ is upper semi-continuous \cite[Chapter 1.1, Corollary 1]{Aubin:Differential}. Since $\mathcal{J}$ is both lower and upper semi-continuous, $\mathcal{J}$ is continuous set-valued map. 
	
	\textit{\textbf{Assertion (b):}} Since $\mathcal{W}$ and $\mathcal{G}$ are non-empty, so is $\mathcal{J}$. Since $\mathcal{J}$ is defined on compact domain and $\mathcal{J}$ is closed, it is compact. Finally, we prove that $\mathcal{J}(x)$ is convex set for any $x\in X$.  $\forall (w^1,y^1),(w^2,y^2)\in\mathcal{J}(x)$ and $\forall \gamma\in(0,1)$, there are
	\begin{eqnarray}
	\notag
	y^1\in \mathcal{G}(x,w^1),y^2\in \mathcal{G}(x,w^2).
	\end{eqnarray}
	Since $\mathcal{G}$ is convex set-valued map (i.e., the graph of $\mathcal{G}$ is convex),
	\begin{eqnarray}
	\notag
	\gamma\mathcal{G}(x,w^1)+(1-\gamma)\mathcal{G}(x,w^2)\subseteq \mathcal{G}(x,\gamma w^1+(1-\gamma)w^2).
	\end{eqnarray}
	Thus
	\begin{eqnarray}
	\label{appen_1}
	\gamma y^1+(1-\gamma)y^2\in \mathcal{G}(x,\gamma w^1+(1-\gamma)w^2).
	\end{eqnarray}
	Since $\mathcal{W}(x)$ is convex,
	\begin{eqnarray}
	\label{appen_2}
	\gamma w^1+(1-\gamma)w^2\in\mathcal{W}(x). 
	\end{eqnarray}
	Thus (\ref{appen_1}) and (\ref{appen_2}) together implies that
	\begin{eqnarray}
	\notag
	(\gamma w^1+(1-\gamma)w^2,\gamma y^1+(1-\gamma)y^2)\in \mathcal{J}(x),
	\end{eqnarray}
	which completes the proof.
\end{proof}

Next, proof of Lemma \ref{lemma:H0} is given.
\begin{proof}[Proof of Lemma \ref{lemma:H0}]
\textcolor{white}{text}

\textit{\textbf{Assertion (a):}} Since Assumption \ref{asmp:SetValuedMap} holds,
$\mathcal{W}(x)$ is non-empty compact set for any $x\in X$, and so is  $\mathcal{G}(x,w)$ for any $x\in X$ and $w\in W$. Thus $\forall x\in X$, $\mathcal{H}_S(x)$ is non-empty, completing the proof of assertion (a).

\textit{\textbf{Assertion (b):}} $\forall (w^1,y^1),(w^2,y^2)\in \mathcal{H}_S(x)$ and $\forall \gamma\in (0,1)$, denote
\begin{eqnarray}
\notag
\hat{w}=\gamma w^1 + (1-\gamma)w^2
\end{eqnarray}
and
\begin{eqnarray}
\notag
\hat{y}=\gamma y^1 + (1-\gamma)y^2.
\end{eqnarray}
Since $\mathcal{W}(x)$ is a convex set as assumed in Assumption \ref{asmp:SetValuedMap} (c), from 
\begin{eqnarray}
\notag
w^1\in \mathcal{W}(x),w^2\in \mathcal{W}(x)
\end{eqnarray}
we have $\hat{w}\in \mathcal{W}(x)$. Since $\mathcal{G}$ is convex set-valued map as assumed in Assumption \ref{asmp:SetValuedMap} (b),
\begin{eqnarray}
\notag
\begin{split}
\hat{y}\ =\ &\gamma y^1 + (1-\gamma)y^2\\
\in\  &\gamma\mathcal{G}(x,w^1)+(1-\gamma)\mathcal{G}(x,w^2)\\
\subseteq&\ \mathcal{G}(x,\lambda w^1+(1-\lambda)w^2) = \mathcal{G}(x,\hat{w}).
\end{split}
\end{eqnarray}
Since $-f_i(x,y,w)$ is concave with respect to $(y,w)$ for any $i\in S$ as assumed in Assumption \ref{asmp:fi} (b),
\begin{eqnarray}
\notag
f_i(x,\hat{y},\hat{w})\le \gamma f_i(x,y^1,w^1)+(1-\gamma)f_i(x,y^2,w^2),\forall i\in S,
\end{eqnarray}
indicating that
\begin{eqnarray}
\notag
\begin{split}
\sum_{i\in S}f_i(x,\hat{y},\hat{w})\le \gamma \sum_{i\in S}f_i&(x,y^1,w^1)\\
&+(1-\gamma)\sum_{i\in S}f_i(x,y^2,w^2).
\end{split}
\end{eqnarray}
Since $(w^1,y^1),(w^2,y^2)\in \mathcal{H}_S(x)$,
\begin{eqnarray}
\notag
\begin{split}
\sum_{i\in S}f_i(x,y^1,w^1)=&\min_{t,v} \sum_{i\in S}f_i(x,v,t)\\
&{\rm s.t.}\ t\in\mathcal{W}(x),v\in \mathcal{G}(x,t)\\
\sum_{i\in S}f_i(x,y^2,w^2)=&\min_{t,v} \sum_{i\in S}f_i(x,v,t)\\
&{\rm s.t.}\ t\in\mathcal{W}(x),v\in \mathcal{G}(x,t)
\end{split}
\end{eqnarray}
Thus
\begin{eqnarray}
\notag
\begin{split}
	\sum_{i\in S}f_i(x,\hat{y},\hat{w})\le&\min_{t,v} \sum_{i\in S}f_i(x,v,t)\\
	&{\rm s.t.}\ t\in\mathcal{W}(x),v\in \mathcal{G}(x,t),
\end{split}
\end{eqnarray}
namely, $(\hat{y},\hat{w})\in \mathcal{H}_S(x)$. So $\mathcal{H}_S(x)$ is a convex set.

\textit{\textbf{Assertion (c):}} Let $f_{S}^*$ denote the optimal objective value of the minimization problem in \eqref{def:HS}. Then $\mathcal{H}_S$ can be reformulated as the intersection of two closed sets:
\begin{eqnarray}
\notag 
\mathcal{H}_S(x)=\mathcal{J}(x)\cap \left\{(w,y)|\sum_{i\in S}f_i(x,y,w)\le f_S^*\right\}
\end{eqnarray}
where $\mathcal{J}(x)$ is defined in \eqref{def:J} and is a compact subsect of $W\times Y$ as stated in Lemma \ref{lemmaJ} (b). 
\begin{eqnarray}
\notag
\left\{(w,y)|\sum_{i\in S}f_i(x,y,w)\le f_S^*\right\}
\end{eqnarray}
is closed since $\sum_{i\in S}f_i(x,y,w)$ is continuous as assumed in Assumption \ref{asmp:fi} (c). Thus $\mathcal{H}_S(x)$ must be a closed subset of $W\times Y$. Moreover, since $W$ and $Y$ are compact, $\mathcal{H}_S(x)$ is compact, which completes the proof of assertion (c).

\textit{\textbf{Assertions (d):}}
Recall the definition of $\mathcal{H}_S$ in \eqref{def:HS}, $\mathcal{H}_S$ is a marginal map. Since $\sum_{i\in S}f_i(x,y,w)$ is continuous and $\mathcal{J}(x)$ is continuous set-valued map with compact values as stated in Lemma \ref{lemmaJ}, the marginal map $\mathcal{H}_S$ is upper semi-continuous according to Lemma \ref{lemma:marginal}. This completes the proof of assertion (d).
\end{proof}

\bibliography{mybib}

\end{document}